\documentclass[letter,11pt]{article}
\usepackage{float,graphicx,verbatim,fullpage,hyperref,amssymb,amsmath,amsthm,amsfonts,enumerate,multicol, xspace,comment,xcolor,thm-restate,url, cite}
\usepackage[T1]{fontenc}
\usepackage[margin=1in]{geometry}
\usepackage{algo}
\usepackage{subfig}
\usepackage[margin=1in]{geometry}
\newcommand*{\email}[1]{\texttt{#1}}
\usepackage{hyperref}

\newtheorem{theorem}{Theorem}[section]
\newtheorem{lemma}{Lemma}[section]
\newtheorem{corollary}{Corollary}[section]

\newtheorem{proposition}{Proposition}[section]
\newtheorem{claim}{Claim}[section]

\newtheorem{observation}{Observation}[section]

\def\final{1}  
\def\iflong{\iffalse}
\ifnum\final=0  
\newcommand{\cnote}[1]{{\color{blue}[{\tiny Calvin: \bf #1}]\marginpar{*}}}
\newcommand{\wnote}[1]{{\color{cyan}[{\tiny Weihang: \bf #1}]\marginpar{*}}}
\newcommand{\knote}[1]{{\color{red}[{\tiny Karthik: \bf #1}]\marginpar{\color{red}*}}}
\newcommand{\todo}[1]{{\color{red}[{\tiny TODO: \bf #1}]\marginpar{\color{red}*}}}
\else 
\newcommand{\cnote}[1]{}
\newcommand{\wnote}[1]{}
\newcommand{\knote}[1]{}
\newcommand{\todo}[1]{}
\fi  

\newcommand{\R}{\mathbb{R}}

\newcommand{\collection}{\mathcal{C}}
\newcommand{\family}{\mathcal{F}}
\newcommand{\mincutfamily}{\mathcal{M}}
\newcommand{\mch}{\mathcal{H}}
\newcommand{\opt}{OPT_k}
\newcommand{\cost}{\text{cost}}

\newcommand{\mcp}{\mathcal{P}}

\newcommand{\deltacard}{d}

\newcommand{\enumhkcut}{\textsc{Enum-Hypergraph-$k$-Cut}\xspace}
\newcommand{\enumgkcut}{\textsc{Enum-Graph-$k$-Cut}\xspace}

\newcommand{\hkcut}{\textsc{Hypergraph-$k$-Cut}\xspace}
\newcommand{\hcut}{\textsc{Hypergraph-MinCut}\xspace}
\newcommand{\gkcut}{\textsc{Graph-$k$-Cut}\xspace}
\newcommand{\gcut}{\textsc{Graph-MinCut}\xspace}
\newcommand{\submodkpartl}{\textsc{Submodular-$k$-Partition}\xspace}

\newcommand{\mypara}[1]{\medskip \noindent {\bf #1}}

\def\complement#1{\overline{#1}}
\def\set#1{\{ #1 \}}
\usepackage[utf8]{inputenc}

\title{Deterministic enumeration of all minimum cut-sets and $k$-cut-sets in hypergraphs for fixed $k$\footnote{University of Illinois, Urbana-Champaign. Email: \email{\{calvinb2, karthe, weihang3\}@illinois.edu}. Supported in part by NSF grants CCF-1814613 and  CCF-1907937.}
}
\author{
Calvin Beideman \and 
Karthekeyan Chandrasekaran \and 
Weihang Wang
}
\date{}

\begin{document}

\maketitle

\begin{abstract}
  We consider the problem of deterministically enumerating all minimum $k$-cut-sets in a given hypergraph for any fixed $k$.
  The input here is a hypergraph $G=(V,E)$ with non-negative hyperedge costs. 
  A subset $F\subseteq E$ of hyperedges is a $k$-cut-set if  
  the number of connected components in $G-F$ is at least $k$ and it is a minimum $k$-cut-set if it has the least cost among all $k$-cut-sets. 
  For fixed $k$, we call the problem of finding a minimum $k$-cut-set as \hkcut and the problem of enumerating all minimum $k$-cut-sets as \enumhkcut. 
  The special cases of \hkcut and \enumhkcut restricted to graph inputs 
  are well-known to be solvable in (randomized as well as deterministic) polynomial time \cite{GH94, KS96, KYN07, Th08}.  
  In contrast, it is only recently that polynomial-time algorithms for \hkcut were developed \cite{CXY19, FPZ19, CC20}. 
  The randomized polynomial-time algorithm for \hkcut that was designed in 2018 \cite{CXY19} showed that the number of minimum $k$-cut-sets in a hypergraph is $O(n^{2k-2})$, where $n$ is the number of vertices in the input hypergraph, and that they can all be enumerated in randomized polynomial time, thus resolving \enumhkcut in randomized polynomial time. A  deterministic polynomial-time algorithm for \hkcut was subsequently designed in 2020 \cite{CC20}, but it is not guaranteed to enumerate all minimum $k$-cut-sets. In this work, we give the first deterministic polynomial-time algorithm to solve \enumhkcut (this is non-trivial even for $k=2$). Our algorithm is based on new structural results that allow for efficient recovery of all minimum $k$-cut-sets by solving minimum $(S,T)$-terminal cuts. Our techniques give new structural insights even for enumerating all minimum cut-sets (i.e., minimum $2$-cut-sets) in a given hypergraph. 
\end{abstract}

\newpage
\setcounter{page}{1}

\section{Introduction}
\label{sec:intro}

A hypergraph $G=(V,E)$ consists of a finite set $V$ of vertices and a
finite set $E$ of hyperedges where each hyperedge $e \in E$ is a subset of $V$. 
We consider 
the problem of enumerating all optimum solutions to the \hkcut problem when $k$ is a fixed constant. 
In \hkcut, the input consists of a hypergraph $G=(V,E)$
with non-negative hyperedge-costs $c: E\rightarrow \R_+$ and a positive integer $k$. The objective is to find a minimum-cost subset of hyperedges whose removal results in at least $k$ connected components. We will call a subset of hyperedges whose removal results in at least $k$ connected components as a \emph{$k$-cut-set}
and a minimum-cost $k$-cut-set as a \emph{minimum $k$-cut-set}; for $k=2$, we will refer to a $2$-cut-set as simply a \emph{cut-set} and a minimum-cost cut-set as a \emph{minimum cut-set}. 
The central problem of interest to this work is that of enumerating all minimum $k$-cut-sets in a given hypergraph with non-negative hyperedge-costs---we will denote this problem as \enumhkcut. Throughout, we will consider $k$ to be a fixed constant integer (e.g., $k=2, 3, 4, ...$). 
We will denote \hkcut and \enumhkcut for graph inputs as \gkcut and \enumgkcut respectively. 
We note that the case of $k=2$ corresponds to global minimum cut which will be discussed shortly. 

\mypara{Partitioning formulation.} 
There is a fundamental structural difference between \hkcut and \gkcut (even for $k=2$), which is especially evident when attempting to enumerate all optimum solutions. In order to illustrate this difference, we discuss an equivalent partitioning formulation of \hkcut.  
In this equivalent formulation, the objective is to find a partition of the vertex set $V$ into $k$ non-empty sets $V_1, V_2, \ldots, V_k$ so as to minimize the cost of hyperedges that cross the partition. A hyperedge $e\in E$ is said to \emph{cross a partition} $V_1, V_2, \ldots, V_k$ 
if it has vertices in at least two parts, that is, there exist distinct $i, j\in [k]$ such that $e\cap V_i\neq \emptyset$ and $e\cap V_j\neq \emptyset$. We will denote a partition of $V$ into $k$ non-empty parts 
as a \emph{$k$-partition} and a $2$-partition as a \emph{cut}. The cost of a $k$-partition is the sum of the cost of hyperedges crossing the partition. A $k$-partition with minimum cost is said to be a \emph{minimum $k$-partition}. 
We will denote the cost of a $2$-partition as its cut value and a minimum $2$-partition as a minimum cut. 

By definition, the number of minimum $k$-cut-sets is at most the number of minimum $k$-partitions. 
Moreover, for a connected graph, the number of minimum $k$-partitions is $O(n^k)$, where $n$ is the number of vertices (i.e., the number of minimum $k$-partitions is polynomial since $k$ is a constant) \cite{KS96, GLL20-STOC, GHLL20}. 
However, for a connected\footnote{A hypergraph is defined to be \emph{connected} if the every cut has at least one hyperedge crossing it.}  hypergraph, the number of minimum $k$-partitions could be exponential while the number of minimum $k$-cut-sets is only polynomial. For example, consider the \emph{spanning-hyperedge-example}: this is the $n$-vertex hypergraph $G=(V,E)$ that consists of only one hyperedge $e$ where $e=V$ with the cost of the hyperedge $e$ being one. 
This hypergraph is connected and has only one minimum $k$-cut-set but $\Theta(k^n)$ minimum $k$-partitions (i.e., an exponential number of minimum $k$-partitions even for $k=2$).  
Thus, if we are hoping for polynomial-time algorithms to enumerate all optimum solutions to \hkcut, then we 
cannot aim to enumerate all minimum $k$-partitions (in contrast to connected graphs). 
This is the reason for defining \enumhkcut as the problem of enumerating all minimum $k$-cut-sets as opposed to enumerating all minimum $k$-partitions.  
For connected graphs, the two definitions are indeed equivalent. 

\medskip
\gkcut for $k=2$ is the minimum cut problem in graphs which is well-known to be solvable in polynomial time. Although the minimum cut problem in graphs has been extensively studied, enumerating all minimum cut-sets in a graph in deterministic polynomial time is already non-trivial. Dinitz, Karzanov, and Lomonosov  \cite{DKL76} constructed a compact representation for all minimum cuts in a connected graph (known as the \emph{cactus representation}) which showed that the number of minimum cuts in a connected graph is at most $\binom{n}{2}$ and that they can all be enumerated in deterministic polynomial time. 
For $k\ge 3$, the number of minimum $k$-partitions in a connected graph is $O(n^k)$---this bound is tight and is a consequence of a recent improved analysis of a random contraction algorithm to solve \gkcut \cite{KS96, GLL20-STOC, GHLL20}; the same random contraction algorithm can also be used to enumerate all minimum $k$-partitions in connected graphs in randomized polynomial time. Deterministic polynomial-time algorithms to enumerate all minimum $k$-partitions in connected graphs are also known. 
We discuss other techniques---both randomized and deterministic---for enumerating minimum cuts and minimum $k$-partitions in graphs in Section \ref{sec:related-work}. 

\hkcut is a natural 
generalization of \gkcut. 
\hkcut for $k=2$ is the minimum cut problem in hypergraphs which is well-known to be solvable in polynomial time \cite{La73}. Once again, enumerating all minimum cut-sets in a hypergraph in deterministic polynomial-time is already non-trivial. There exists a compact representation of all minimum cut-sets in a hypergraph \cite{ChekuriX18}---namely the  \emph{hypercactus representation}---which also implies that the number of minimum cut-sets in a hypergraph is at most $\binom{n}{2}$ and that they can all be enumerated in deterministic polynomial time. To the best of the authors' knowledge, this is the only known technique for efficient deterministic enumeration of all minimum cut-sets in a hypergraph. 

\hkcut is a special case of \submodkpartl 
(e.g., see \cite{Zhthesis, ZNI05, OFN12, CC20}).  
Owing to this connection, 
the complexity of \hkcut 
for any fixed $k\ge 3$ 
has been an intriguing open question until recently. 
A randomized polynomial-time algorithm for \hkcut was designed in 2018 by Chandrasekaran, Xu, and Yu \cite{CXY19}. The analysis of this algorithm showed that the number of minimum $k$-cut-sets is $O(n^{2k-2})$, where $n$ is the number of vertices in the input hypergraph (i.e., the number of minimum $k$-cut-sets is polynomial), and that they can all be enumerated in randomized polynomial time (also see \cite{FPZ19}). 
Subsequently, Chandrasekaran and Chekuri designed a deterministic polynomial-time algorithm for \hkcut in 2020 \cite{CC20}. 
However, their deterministic algorithm is guaranteed to identify only one minimum $k$-cut-set and not all. The next natural question is whether all minimum $k$-cut-sets can be enumerated in deterministic polynomial time---namely, can we solve \enumhkcut in deterministic polynomial time?

As mentioned earlier, the only known technique for \enumhkcut for $k=2$ is via the hypercactus representation which does not seem to generalize to $k\ge 3$ (in fact, it is unclear if cactus representation generalizes to $k\ge 3$ even in graphs). 
Moreover, all deterministic techniques for \enumgkcut address the problem of enumerating all minimum $k$-partitions in connected graphs---see Section \ref{sec:related-work}; hence, all these techniques fail for \enumhkcut (as seen from the spanning-hyperedge-example). For hypergraphs, we necessarily have to work with minimum $k$-cut-sets as opposed to minimum $k$-partitions. 
Working with minimum $k$-cut-sets as opposed to minimum $k$-partitions in the deterministic setting is a 
technical challenge that has not been undertaken in any of the previous works (even for graphs). 
We overcome this technical challenge in this work.
We adapt Chandrasekaran and Chekuri's deterministic approach for \hkcut and augment it with structural results for minimum $k$-cut-sets to prove our main result stated below. 

\begin{theorem}\label{thm:enumhkcut-algo}
There is a deterministic polynomial-time algorithm for \enumhkcut for every fixed $k$. 
\end{theorem}

Although we chose to highlight the above algorithmic result in this introduction, we emphasize that the structural theorems that form the backbone of the algorithmic result  are our main technical contributions (see Theorems \ref{theorem: structure thm 1} and \ref{theorem: structure thm 2}). 
We discuss these structural theorems in the technical overview section. 
By tightening the proof technique of one of our structural theorems for $k=2$, we obtain an arguably elegant structural explanation for the number of minimum cut-sets in a hypergraph being at most $\binom{n}{2}$---see Theorem \ref{theorem: structure thm 2-for-min-cut}. Theorem \ref{theorem: structure thm 2-for-min-cut} leads to an alternative deterministic polynomial-time algorithm to enumerate all minimum cut-sets in a hypergraph (that is relatively simpler than computing a  hypercactus representation). 
We believe that our structural theorems are likely to be of independent interest.

\subsection{Technical overview and main structural results}
\label{sec:techniques}
We focus on the unit-cost variant of \enumhkcut in the rest of this work for the sake of notational simplicity. Throughout, we will allow
multigraphs and hence, this is without loss of generality. Our
algorithms extend in a straightforward manner to arbitrary hyperedge
costs. They rely only on minimum $(s,t)$-terminal cut computations and hence, they are strongly polynomial-time algorithms.

A key algorithmic tool will be the use of terminal cuts.  We need some
notation. Let $G=(V,E)$ be a hypergraph. Throughout this work, 
$n$ will denote the number of vertices in $G$ and $p:=\sum_{e\in E}|e|$ will denote the representation size of $G$. 
We will denote a partition of the vertex set into $h$ non-empty parts by an ordered tuple $(V_1, \ldots, V_h)$. For a non-empty proper subset $U$ of vertices,
we will use $\complement{U}$ to denote $V\setminus U$, $\delta(U)$
to denote the set of hyperedges crossing the $2$-partition $(U, \complement{U})$, 
and $\deltacard(U):=|\delta(U)|$.   
We recall that $\delta(U)=\delta(\complement{U})$, so we will use $\deltacard(U)$ to denote the cost of the cut $(U, \complement{U})$. 
More generally, given a partition
$\mcp=(V_1,V_2,\ldots,V_h)$, we denote 
the 
set 
of hyperedges crossing the partition by 
$\delta(V_1, V_2, \ldots, V_h)$ 
(also by $\delta(\mcp)$ for brevity)
and the number of hyperedges crossing the partition by 
$\cost(V_1, V_2, \ldots, V_h) := |\delta(V_1, V_2, \ldots, V_h)|$
(also by $\cost(\mcp)$ for brevity). 
Let $S$, $T$ be
disjoint non-empty subsets of vertices. A $2$-partition $(U, \complement{U})$ is
an $(S,T)$-terminal cut if $S\subseteq U\subseteq V\setminus T$. Here,
the set $U$ is known as the source set and the set $\complement{U}$ is
known as the sink set.  A minimum-cost $(S,T)$-terminal cut is known
as a \emph{minimum $(S,T)$-terminal cut}.  Since there could be multiple
minimum $(S,T)$-terminal cuts, we will be interested in \emph{source
  minimal} minimum $(S,T)$-terminal cuts 
  and \emph{source maximal} minimum $(S,T)$-terminal cuts. 
For every pair of disjoint non-empty subsets $S$ and $T$ of vertices, there exists a unique source minimal minimum $(S,T)$-terminal cut and it can be found
in deterministic polynomial time via standard maxflow algorithms; a similar result holds for source maximal minimum $(S,T)$-terminal cuts. 

Our algorithm is inspired by the divide and conquer approach introduced by Goldschmidt and Hochbaum for \gkcut \cite{GH94}. This approach was generalized by Kamidoi, Yoshida, and Nagamochi to solve \enumgkcut \cite{KYN07} and by Chandrasekaran and Chekuri to solve \hkcut \cite{CC20}, both in deterministic polynomial time. 
The techniques of \cite{GH94} and \cite{KYN07} are not applicable to \enumhkcut since they are tailored to graphs and do not extend to hypergraphs. 
We describe the details of the divide and conquer approach for \hkcut due to Chandrasekaran and Chekuri \cite{CC20}. The goal here is to identify one part of some fixed minimum $k$-partiton $(V_1, V_2, \ldots, V_k)$, say $V_1$ without loss of generality, and then recursively find a minimum $(k-1)$-partition in the subhypergraph $G[\complement{V_1}]$, where $G[\complement{V_1}]$ is the hypergraph obtained from $G$ by discarding the vertices in $V_1$ and by discarding all hyperedges 
that intersect $V_1$. Now, how does one find such a part $V_1$? Chandrasekaran and Chekuri proved a key structural theorem for this: Suppose $(V_1, \ldots, V_k)$ is a \emph{$V_1$-maximal} minimum $k$-partition---i.e., there is no other minimum $k$-partition $(V_1', \ldots, V_k')$ such that $V_1$ is a proper subset of $V_1'$. Then, they showed that for every subset $T\subseteq \complement{V_1}$ such that $T\cap V_j\neq \emptyset$ for all $j\in \{2, \ldots, k\}$, there exists a subset $S\subseteq V_1$ of size at most $2k-2$ such that $(V_1, \complement{V_1})$ is the source maximal minimum $(S,T)$-terminal cut. A consequence of this structural theorem is that if we compute the collection $\collection$ consisting 
of the source side of the source maximal minimum $(S,T)$-terminal cut
for all possible pairs $(S,T)$ of disjoint subsets of vertices $S$ and $T$ with $|S|\le 2k-2$ and $|T|\le k-1$, then the set $V_1$ will be in this collection $\collection$ (by applying the structural theorem to a set $T$ of size $k-1$ with $|T\cap V_j|=1$ for all $j\in \{2, \ldots, k\}$). Moreover, the size of the collection $\collection$ is only $O(n^{3k-3})$. Hence, recursing on $G[\complement{U}]$ for each set $U$ in the collection $\collection$ will identify a minimum $k$-partition within a total run-time of $n^{O(k^2)}$ source maximal minimum $(S,T)$-terminal cut computations. 

The limitation of the structural theorem of Chandrasekaran and Chekuri \cite{CC20} is that it aims to recover a minimum $k$-partition and in particular, a $V_1$-maximal minimum $k$-partition. For the purposes of enumerating all minimum $k$-cut-sets, this is insufficient as we have seen from the spanning-hyperedge-example. In particular, their structural theorem cannot be used to even enumerate all minimum cut-sets in a hypergraph. 
We prove two structural theorems that will help in enumerating minimum $k$-cut-sets. We describe these structural theorems now. 

Our goal is to deterministically enumerate a polynomial-sized family $\family$ of $k$-cut-sets such that $\family$ contains all minimum $k$-cut-sets. Let $F$ be an arbitrary minimum $k$-cut-set. 
Since $F$ is a minimum $k$-cut-set, there exists a minimum $k$-partition $(V_1, \ldots, V_k)$ such that $F=\delta(V_1, \ldots, V_k)$. We note that $d(V_1)\le |F|$ by definition of the hypergraph cut function $d:2^V\rightarrow \R$. 
We distinguish two cases:

\medskip
\noindent \textbf{Case 1.} Suppose $d(V_1)<|F|$. 
In order to identify minimum $k$-cut-sets $F$ that have this property, 
we show the following structural theorem. 
    \begin{restatable}{theorem}{thmStructureOne}
\label{theorem: structure thm 1}
Let $G=(V,E)$ be a hypergraph and let $\opt$ be the value of a minimum $k$-cut-set in $G$ for some integer $k\ge 2$. Suppose $(U,\complement{U})$ is a $2$-partition of $V$ with $d(U)<\opt$. Then, for every pair of vertices $s\in U$ and $t\in\overline{U}$, there exist subsets  $S\subseteq U\setminus \{s\}$ and $T\subseteq \complement{U}\setminus \{t\}$ with $|S|\le 2k-3$ and $|T|\leq 2k-3$ such that $(U,\complement{U})$ is the unique minimum $(S\cup\{s\},T\cup\{t\})$-terminal cut in $G$.
\end{restatable}

The advantage of this structural theorem is that it allows for a recursive approach to enumerate a polynomial-sized family of minimum $k$-cut-sets containing $F$ under the assumption that $d(V_1)<|F|=\opt$ (similar to the approach of Chandrasekaran and Chekuri).

The drawback of this structural theorem is that it only addresses the case of $d(V_1)<|F|$. It is possible that the minimum $k$-cut-set $F$ satisfies $d(V_1)=|F|$. For example, consider the problem of enumerating all minimum cut-sets in a hypergraph (i.e., \enumhkcut for $k=2$)---Theorem \ref{theorem: structure thm 1} does not help in this case since there will be no cut $(U, \complement{U})$ with $d(U)<\text{OPT}_2$. This motivates the second case. 

\medskip
\noindent \textbf{Case 2.} Suppose $d(V_1) = |F|$. 
In this case, we need to enumerate a polynomial-sized family of $k$-cut-sets containing $F$, but we cannot hope to enumerate all minimum $k$-partitions $(V_1', \ldots, V_k')$ for which $F=\delta(V_1', \ldots, V_k')$ 
(e.g., again consider the spanning-hyperedge-example for $k=2$ for which the unique minimum cut-set $F$ has $|F|=d(V_1)$ for exponentially many minimum cuts $(V_1, V_2)$ and hence, we cannot hope to enumerate all minimum cuts). 
We observe that if $d(V_1)=|F|$, then the set $F$ of hyperedges should be equal to the set of hyperedges crossing $(V_1, \complement{V_1})$, i.e., $\delta(V_1) = F=\delta(V_1, \ldots, V_k)$. 
We show the following structural theorem to exploit this observation. 

\begin{restatable}{theorem}{thmStructureTwo}
\label{theorem: structure thm 2}
Let $G=(V,E)$ be a hypergraph, $k \geq 2$ be an integer, and $\mathcal{P}=(V_1,\ldots,V_k)$ be 
a minimum $k$-partition such that $\delta(V_1)=\delta(\mathcal{P})$. Then, for all subsets $T\subseteq \complement{V_1}$ such that $T\cap V_j\neq\emptyset$ for all $j\in\{2,3,\ldots,k\}$, there exists a subset $S\subseteq V_1$ with $|S|\leq 2k-1$ such that the source minimal minimum $(S,T)$-terminal cut $(A,\complement{A})$ satisfies $\delta(A)=\delta(V_1)$ and $A\subseteq V_1$.
\end{restatable}

We recall that for fixed disjoint subsets $S, T\subseteq V$, the source minimal minimum $(S,T)$-terminal cut is unique. 
We emphasize the main feature of Theorem \ref{theorem: structure thm 2}: it aims to recover 
only the hyperedges crossing the cut $(V_1, \complement{V_1})$ 
but not the cut $(V_1, \complement{V_1})$ itself. It shows the existence of a small-sized witness which allows us to recover $\delta(V_1)$---namely a pair $(S,T)$ with $|S|, |T|=O(k)$ for which $\delta(V_1)$ is the cut-set of the source minimal minimum $(S,T)$-terminal cut. 
In this sense, Theorem \ref{theorem: structure thm 2} addresses the drawback of Theorem \ref{theorem: structure thm 1}. 

Theorems \ref{theorem: structure thm 1} and \ref{theorem: structure thm 2} can be used to design a recursive algorithm that enumerates all minimum $k$-cut-sets in deterministic polynomial time (along the lines of the algorithm of Chandrasekaran and Chekuri described above). Here, we describe a more straightforward non-recursive deterministic polynomial-time algorithm. 
For each pair of subsets of vertices $S, T$ of size at most $2k-1$, we compute the source minimal minimum $(S, T)$-terminal cut $V_{S, T}$; if $G-\delta(V_{S, T})$ has at least $k$ connected components, then we add $\delta(V_{S, T})$ to the candidate family $\family$; otherwise, we add $V_{S, T}$ to the collection $\collection$. Next, we consider all possible $k$-partitions $(U_1, \ldots, U_k)$ of the vertex set where all sets $U_1, \ldots, U_k$ are in the collection $\collection$ and add the set $\delta(U_1, \ldots, U_k)$ of hyperedges to the family $\family$. 
We now sketch the argument to show that the family $\family$ contains the (arbitrary) minimum $k$-cut-set $F$. Recall that there exists a minimum $k$-partition $(V_1, \ldots, V_k)$ such that $F$ is the set of hyperedges crossing this $k$-partition, i.e., $F=\delta(V_1, \ldots, V_k)$. We have two possibilities: (1) if $d(V_i)<|F|$ for every $i\in [k]$, then by Theorem \ref{theorem: structure thm 1}, every set $V_i$ is in the collection $\collection$ (by applying Theorem \ref{theorem: structure thm 1} to $(U=V_i, \complement{U}=\complement{V_i})$ and arbitrary vertices $s\in V_i, t\in \complement{V_i}$), and hence $F\in \family$; (2) if $d(V_i)=|F|$  for some $i\in [k]$, then by Theorem \ref{theorem: structure thm 2}, one of the sets $V_{S, T}$ has $\delta(V_{S, T})=\delta(V_i)=F$ and hence, once again $F\in \family$. We can prune the family $\family$ to return the subfamily of minimum $k$-cut-sets in it. The size of the collection $\collection$ is $O(n^{4k-2})$ and the size of the family $\family$ is $O(n^{4k^2})$. The run-time is $O(n^{4k-2})T(n, p)+O(n^{4k^2})$, where $T(n,p)$ is the time complexity for computing the source minimal minimum $(s,t)$-terminal cut in a $n$-vertex hypergraph of size $p$. 

\mypara{Additional consequence of Theorem \ref{theorem: structure thm 2}.} Theorem \ref{theorem: structure thm 2} is the technical novelty of this work. We emphasize another structural consequence of Theorem \ref{theorem: structure thm 2} by using it to bound the number of minimum cut-sets in a hypergraph. 
Let $t$ be an arbitrary vertex in the hypergraph $G=(V,E)$. 
Consider the sets 
\begin{align*}
\mathcal{H}&:=\{U\subseteq V\setminus \{t\}: (U, \complement{U}) \text{ is a minimum cut in }G\} \text{ and}\\
\mincutfamily &:= \{\delta(U): U\in \mathcal{H}\}.
\end{align*}
We note that $\mincutfamily$ is the family of all minimum cut-sets in the hypergraph. By applying Theorem \ref{theorem: structure thm 2} for $k=2$ and $T=\{t\}$, we obtain that for every set $U\in \mch$, there exists a subset $S\subseteq U$ with $|S|\le 3$ such that the source minimal minimum $(S, \{t\})$-terminal cut $(A, \complement{A})$ satisfies $\delta(A) = \delta(U)$. Consequently, the size of the set $\mincutfamily$ is at most the number of possible ways to choose a non-empty subset $S\subseteq V\setminus \{t\}$ of size at most $3$ which is $\binom{n-1}{1}+\binom{n-1}{2}+\binom{n-1}{3}=O(n^3)$, where $n:=|V|$. Thus, we have concluded that the number of minimum cut-sets in a $n$-vertex hypergraph is $O(n^3)$. 

We recall that the number of minimum cut-sets in a $n$-vertex hypergraph is known to be at most $\binom{n}{2}$ \cite{ChekuriX18, GKP17}. So, the $O(n^3)$ upper bound on the number of minimum cut-sets that we obtained above based on Theorem \ref{theorem: structure thm 2} appears to be weak. We show the following strengthening of Theorem \ref{theorem: structure thm 2} for $k=2$ to get the tighter bound.

\begin{restatable}{theorem}{thmStructureTwoForMinCut}
\label{theorem: structure thm 2-for-min-cut}
Let $G=(V,E)$ be a hypergraph and $\mathcal{P}=(V_1,V_2)$ be 
a minimum cut. Then, for all non-empty subsets $T\subseteq V_2$, there exists a subset $S\subseteq V_1$ with $|S|\leq 2$ such that the source minimal minimum $(S,T)$-terminal cut $(A,\complement{A})$ satisfies $\delta(A)=\delta(V_1)$ and $A\subseteq V_1$. 
\end{restatable}

By applying Theorem \ref{theorem: structure thm 2-for-min-cut} for $T=\{t\}$, we obtain that for every set $U\in \mch$, there exists a subset $S\subseteq U$ with $|S|\le 2$ such that the source minimal minimum $(S,\{t\})$-terminal cut $(A, \complement{A})$ satisfies $\delta(A) = \delta(U)$. Hence, the size of the set $\mincutfamily$ is at most the number of possible ways to choose a non-empty subset $S\subseteq V\setminus \{t\}$ of size at most $2$ which is $\binom{n-1}{1} + \binom{n-1}{2} = \binom{n}{2}$. Thus, we have obtained a structural explanation (based on Theorem \ref{theorem: structure thm 2-for-min-cut}) for the number of minimum cut-sets in a hypergraph being at most $\binom{n}{2}$. 
Theorem \ref{theorem: structure thm 2-for-min-cut} can also be used to enumerate all minimum cut-sets in a given hypergraph using $\binom{n}{2}$ source minimal minimum $(S,T)$-terminal cut computations. 

Theorem \ref{theorem: structure thm 2-for-min-cut} should be compared with a similar-looking structural theorem for graphs that was shown by Goemans and Ramakrishnan \cite{GR95}. Goemans and Ramakrishnan showed that (Theorem 15 in \cite{GR95}) if $G$ is a connected graph, then for every set $U\in \mch$, there exists a subset $S\subseteq V_1$ with $|S|\le 2$ such that $(U, \complement{U})$ is the source minimal minimum $(S, \{t\})$-terminal cut. This leads to a structural explanation for the number of minimum cuts in a connected graph being at most $\binom{n}{2}$. Our Theorem \ref{theorem: structure thm 2-for-min-cut} can be seen as a counterpart of Goemans and Ramakrishnan's result for hypergraphs, but it differs from their result in two aspects: (1) their result does not hold for hypergraphs---the number of minimum cuts in a connected hypergraph could be exponential as we have seen from the spanning-hyperedge-example and (2) the proof of their result is based on the \emph{submodular triple inequality} which holds only for the graph cut function but fails for the hypergraph cut function. So, our Theorem \ref{theorem: structure thm 2-for-min-cut} is more general as it handles minimum cut-sets in hypergraphs and moreover, needs a different proof technique compared to \cite{GR95}. We mention that Goemans and Ramakrishnan's result for connected graphs was our inspiration for Theorem \ref{theorem: structure thm 2-for-min-cut}, which in turn, was our starting point for Theorem \ref{theorem: structure thm 2}. 

\mypara{Organization.} We discuss special cases of \enumhkcut that have been addressed in the literature in Section \ref{sec:related-work}. 
In Section \ref{sec:prelims}, we recall properties of the hypergraph cut function that will be useful to prove our structural theorems. This section 
contains a strengthening of a partition uncrossing theorem from \cite{CC20} whose proof appears in Appendix \ref{section:uncrossing-for-hypergraph-cut-function}. In Section \ref{sec:enumeration-algorithm}, we formally describe and analyze the deterministic polynomial-time algorithm for \enumhkcut that utilizes our two structural theorems (Theorems \ref{theorem: structure thm 1} and \ref{theorem: structure thm 2}). We prove Theorems \ref{theorem: structure thm 1} and \ref{theorem: structure thm 2} in Sections \ref{sec:structure-1} and \ref{sec:structure-2} respectively. 
We prove the strengthening of Theorem \ref{theorem: structure thm 2} for $k=2$---namely Theorem \ref{theorem: structure thm 2-for-min-cut}---in Section \ref{sec:stronger-structure-2-for-k-equals-2}. We conclude with a few open problems in Section \ref{sec:conclusion}.

\subsection{Related work}
\label{sec:related-work}

\medskip
In this section, we discuss 
known techniques for the enumeration problem in the special case of $k=2$ and the special case of graphs along with challenges involved in adapting these techniques to hypergraphs for $k\ge 3$. 

\mypara{\enumgkcut for $k=2$.} \gkcut for $k=2$ is the global minimum cut problem (denoted \gcut) which has been extensively studied.
However, 
efficient deterministic enumeration of all minimum cut-sets in a given connected graph 
is already non-trivial. 
Dinitz, Karzanov, and Lomonosov \cite{DKL76} showed that the number of minimum cuts in a connected graph is at most $\binom{n}{2}$, where $n$ is the number of vertices in the input graph, and they can all be enumerated in deterministic polynomial time. In particular, they designed a compact data structure, namely a cactus graph, to  represent all minimum cuts in a connected graph. 
The upper bound of $\binom{n}{2}$ on the number of minimum cuts in a connected graph is tight as illustrated by the cycle-graph on $n$ vertices.  
Using the seminal random contraction technique, Karger \cite{Kar93} showed a stronger result that the number of $\alpha$-approximate minimum cuts in a connected graph is  $O(n^{2\alpha})$ and they can all be enumerated in randomized polynomial time for constant $\alpha$. 
Karger's tree packing technique \cite{Karger00} also leads to a deterministic polynomial-time algorithm to enumerate all $\alpha$-approximate minimum cuts in a connected graph for constant $\alpha$. 
Nagamochi, Nishimura, and Ibaraki \cite{NNI97} tightened Karger's bound for a particular value of $\alpha$ via the edge splitting operation: the number of  $(4/3-\epsilon)$-approximate minimum cuts in a connected graph is at most $\binom{n}{2}$ for any $\epsilon>0$. This fact was also shown by Goemans and Ramakrishnan \cite{GR95} via a structural result (see discussion after Theorem \ref{theorem: structure thm 2-for-min-cut} above). Henzinger and Williamson \cite{HW96} extended Nagamochi, Nishimura, and Ibaraki's edge splitting technique to show that the number of $(3/2-\epsilon)$-approximate minimum cuts in a connected graph is  $O(n^2)$ for any $\epsilon>0$. The results of Nagamochi, Nishimura, and Ibaraki, Goemans and Ramakrishnan, and Henzinger and Williamson are all constructive and deterministic (i.e., lead to deterministic polynomial-time algorithms to enumerate the respective approximate minimum cuts) and they bound the number of minimum cuts in a connected graph (as opposed to minimum cut-sets). 

\noindent \emph{Polynomial-delay algorithms.} An alternative line of work aims to enumerate \emph{all} cuts in hypergraphs in non-decreasing order of cut value with polynomial time delay between outputs. Such algorithms are known as \emph{polynomial-delay} algorithms in the literature. 
Polynomial-delay algorithms have been designed based on polynomial-time solvability of minimum $(s,t)$-terminal cut and using the Lawler-Murty schema \cite{HPQ84,VY92,NI-book, AMMQ15}. 
Since we know that the number of minimum cuts in a connected graph is polynomial, the existence of a polynomial-delay algorithm immediately implies a polynomial-time algorithm to solve \enumgkcut for $k=2$. This approach does not extend to \enumhkcut for $k=2$ since the number of minimum cuts in a hypergraph can be exponential (e.g., recall the spanning-hyperedge-example).

\mypara{\enumhkcut for $k=2$.} \hkcut for $k=2$ is the global minimum cut problem (denoted \hcut) which has also been extensively studied. We note that the number of minimum cuts in a connected hypergraph could be exponential (e.g., consider the spanning-hyperedge-example). But, how about the number of minimum cut-sets? The number of minimum cut-sets in a hypergraph is at most $\binom{n}{2}$ via decomposition theorems of Cunningham and Edmonds \cite{CE80}, Fujishige \cite{Fuj83}, and Cunningham \cite{Cun80} on submodular functions. Cheng \cite{Che99} designed an explicit \emph{hypercactus representation} for all minimum cut-sets in a hypergraph. Chekuri and Xu \cite{ChekuriX18} designed a faster deterministic polynomial-time algorithm to obtain a hypercactus representation 
(along with all minimum cut-sets) 
of a given hypergraph. Ghaffari, Karger, and Panigrahi \cite{GKP17} (also see \cite{CXY19, FPZ19}) introduced a random contraction technique to solve \hcut which also implied that the number of minimum cut-sets in a hypergraph is at most $\binom{n}{2}$ and that they can all be enumerated in randomized polynomial time. 

We mention that 
in contrast to graphs, the number of constant-approximate minimum cut-sets in a hypergraph can be exponential. 
In fact, the number of $(1+\epsilon)$-approximate minimum cut-sets in a connected hypergraph can be exponential\footnote{Consider the $n$-vertex hypergraph $G=(V,E)$ where $E$ consists of all size-$2$ hyperedges each of cost $\delta=\epsilon(\binom{n}{2}-(1+\epsilon)(n-1))^{-1}$ and a hyperedge $e=V$ of cost $1$. The cost of a minimum cut is $\lambda:=1+\delta (n-1)$. The cost of every cut is at most $1+\delta \binom{n}{2}\le (1+\epsilon)\lambda$. } for any $\epsilon\in (0,1)$. 
Moreover, the techniques of Nagamochi, Nishimura, and Ibaraki, Goemans and Ramakrishnan, and Henzinger and Williamson even when restricted to minimum cuts (as opposed to approximate minimum cuts) cannot extend to hypergraphs: This is because, their techniques are tailored to enumerate all minimum cuts in a connected graph as opposed to all minimum cut-sets; we have already seen that the spanning-hyperedge-example has exponential number of minimum cuts and hence, 
all of them cannot be enumerated in polynomial time. 

\smallskip
\noindent 
\emph{Multiterminal variants for $k$-cut:} We mention that \gkcut and \hkcut have natural variants involving separating specified
terminal vertices $s_1,s_2,\ldots,s_k$. These variants are NP-hard for $k\ge 3$ even in graphs and hence, these variants are not viable lines of attack for \gkcut and \hkcut. We refer the reader to  \cite{CC20} for a discussion of approximation algorithms for these variants.

\mypara{\enumgkcut.} 
\gkcut for $k\ge 3$ has 
a rich literature 
with substantial recent work \cite{GH94, KS96, Th08, KYN07, SV95, RS08, Ma18, GLL18-SODA, GLL18-FOCS, GLL19-STOC, GLL20-STOC, GHLL20, LSS20, CQX20}. 
Goldschmidt and Hochbaum (1988) \cite{GH94} initiated the study on \gkcut by showing that it is NP-hard when $k$ is part of the input and that it is polynomial-time solvable when $k$ is any fixed constant (polynomial-time solvability is not obvious even for $k=3$). 
Recall that we consider $k$ to be a fixed constant throughout this work. 
Goldschmidt and Hochbaum introduced a 
divide-and-conquer approach for \gkcut which resulted in a deterministic polynomial-time algorithm. 
However, their result did not guarantee any bound on the number of minimum $k$-partitions or minimum $k$-cut-sets in connected graphs. Karger and Stein \cite{KS96} gave a randomized polynomial-time algorithm for \gkcut via the random contraction technique. In addition, they showed that the number of minimum $k$-partitions in a connected graph is $O(n^{2k-2})$ and they can all be enumerated in randomized polynomial time. The bound on the number of minimum $k$-partitions in a connected graph 
has recently been improved to 
$O(n^k)$ 
\cite{GLL20-STOC,GHLL20}. We mention that the upper bound of $O(n^k)$ on the number of minimum $k$-partitions in a connected graph is tight as illustrated by the cycle-graph on $n$ vertices. 

There are two known approaches to solve \enumgkcut in deterministic polynomial time: (1) Thorup \cite{Th08} showed that the tree packing approach can be used to obtain a polynomial-time algorithm for \gkcut; this approach also extends to solve \enumgkcut (also see \cite{CQX20}). 
(2) Kamidoi, Yoshida, and Nagamochi \cite{KYN07} extended Goldschmidt and Hochbaum's divide and conquer approach to solve \enumgkcut. 

\mypara{\hkcut.} 
The complexity of \hkcut was open since the work of Goldschmidt and Hochbaum for \gkcut (1988) \cite{GH94} until recently. 
Although certain special cases of \hkcut were known to be solvable in polynomial time \cite{F10, X10}, considerable progress on \hkcut happened only in the last $3$ years. 
Chandrasekaran, Xu, and Yu (2018) \cite{CXY19}
designed the first randomized polynomial-time algorithm for \hkcut; 
their Monte Carlo algorithm runs in $\tilde{O}(pn^{2k-1})$ time where $p = \sum_{e \in E} |e|$ is the representation size of the input hypergraph. 
Fox, Panigrahi, and Zhang \cite{FPZ19} 
improved the randomized run-time to $\tilde{O}(mn^{2k-2})$, where $m$ is the number of hyperedges in the input hypergraph. 
Both these randomized
algorithms are based on random contraction of hyperedges and are
inspired partly by earlier work in \cite{GKP17} for \hcut. These randomized algorithms also imply that the number of minimum $k$-cut-sets is $O(n^{2k-2})$ and that all of them can be enumerated in randomized polynomial time. Chandrasekaran and Chekuri (2020) \cite{CC20} designed a deterministic polynomial-time algorithm for \hkcut via a divide and conquer approach. 
We emphasize that their algorithm finds a minimum $k$-partition and did not have the tools to find all minimum $k$-cut-sets. 

A polynomial bound on the number of minimum $k$-cut-sets along with the existence of a randomized polynomial-time algorithm to enumerate all of them raises the possibility of a deterministic algorithm for \enumhkcut. 
As we mentioned earlier, there are two deterministic  approaches for \enumgkcut---tree packing and divide-and-conquer. The tree packing approach does not seem to extend to hypergraphs (even for \hcut). This leaves the divide-and-conquer approach. 
Notably, this approach also led to the first deterministic algorithm for \hkcut in the work of Chandrasekaran and Chekuri \cite{CC20}. 
As mentioned earlier, we adapt Chandrasekaran and Chekuri's divide-and-conquer approach and augment it with structural results for minimum $k$-cut-sets to prove our main result stated in Theorem \ref{thm:enumhkcut-algo}. 

\subsection{Preliminaries} \label{sec:prelims}
Let $G=(V,E)$ be a hypergraph. Throughout, we will follow the notation mentioned in the second paragraph of Section \ref{sec:techniques}. 
We will repeatedly rely on the fact that the hypergraph cut function $d:2^V\rightarrow \R_+$ is symmetric and submodular. We recall that a set function $f:2^V\rightarrow \R$ is \emph{symmetric} if $f(U)=f(\complement{U})$ for all $U\subseteq V$ and is \emph{submodular} if $f(A) + f(B) \ge f(A\cap B)+f(A\cup B)$ for all subsets $A, B\subseteq V$. 

We will need a partition uncrossing theorem that is a strengthening of a result from \cite{CC20}. We state the strengthened version below. 
See Figure
\ref{figure:uncrossing} for an illustration of the sets that appear in the statement of Theorem \ref{theorem:hypergraph-uncrossing}. 
We emphasize that the second conclusion in the statement of Theorem \ref{theorem:hypergraph-uncrossing} is the strengthening. The proof of the second conclusion is similar to the proof of the first conclusion which appears in \cite{CC20}---we present a proof of both conclusions for the sake of completeness in 
Appendix 
\ref{section:uncrossing-for-hypergraph-cut-function}. 

\begin{restatable}{theorem}{thmHypergraphUncrossing}
\label{theorem:hypergraph-uncrossing}
Let $G=(V,E)$ be a hypergraph, $k\ge 2$ be an integer and
$\emptyset\neq R\subsetneq U\subsetneq V$. Let
$S=\{u_1,\ldots, u_p\}\subseteq U\setminus R$ for $p\ge 2k-2$. Let
$(\complement{A_i}, A_i)$ be a minimum
$((S\cup R)\setminus \set{u_i}, \complement{U})$-terminal cut. Suppose
that $u_i\in A_i\setminus (\cup_{j\in [p]\setminus \set{i}}A_j)$ for
every $i\in [p]$.  Then, the following two hold: 
\begin{enumerate}
    \item There exists a $k$-partition
$(P_1, \ldots, P_k)$ of $V$ with $\complement{U}\subsetneq P_k$ such
that
\[
\cost(P_1, \ldots, P_k) \le \frac{1}{2}\min\{\deltacard(A_i) + \deltacard(A_j): i, j\in [p], i\neq j\}.
\]
\item 
Moreover, if there exists a hyperedge 
$e\in E$ such that 
$e$ intersects $W:=\cup_{1\le i<j\le p}(A_i\cap A_j)$, $e$ intersects $Z:=\cap_{i\in [p]} \complement{A_i}$, and $e$ is contained in $W\cup Z$, 
then the inequality 
in the previous conclusion 
is strict.
\end{enumerate}
\end{restatable}

\begin{figure}[htb]
\centering
\includegraphics[width=0.6\textwidth]{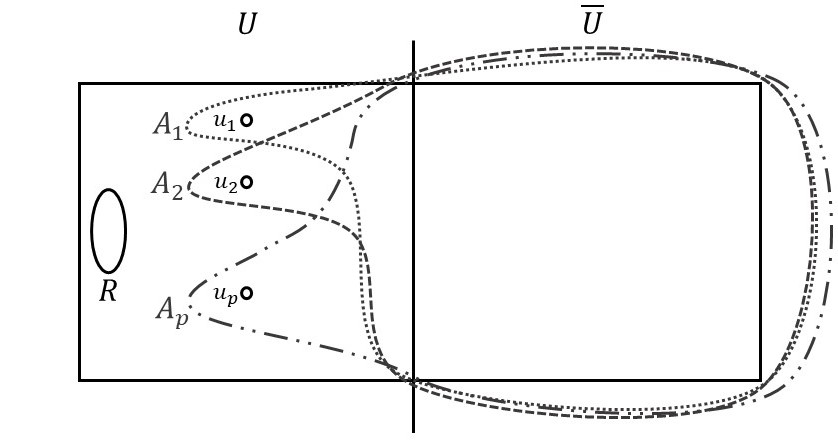}
\caption{Illustration of the sets that appear in the statement of Theorem \ref{theorem:hypergraph-uncrossing}.}
\label{figure:uncrossing}
\end{figure}

\section{Enumeration Algorithm} \label{sec:enumeration-algorithm}

We will use Theorems \ref{theorem: structure thm 1} and \ref{theorem: structure thm 2} to design a deterministic polynomial-time algorithm for \enumhkcut in this section. 
We describe the formal algorithm in Figure \ref{Algo: enum-cuts}. It enumerates $n^{O(k)}$ source minimal minimum $(S,T)$-terminal cuts and considers the cut-set crossing each cut in this collection. If the removal of the cut-set leads to at least $k$ connected components, then it adds such a cut-set to the candidate family $\family$; otherwise, it adds the source set of the cut into a candidate collection $\collection$. 
Next, the algorithm considers 
all possible $k$-partitions that can be formed using the sets in the collection $\collection$ and adds the set of hyperedges crossing the $k$-partition to the family $\family$. Finally, it prunes the family $\family$ to return all minimum $k$-cut-sets in it.  
The run-time guarantee and the cardinality of the family of $k$-cut-sets returned by the algorithm are given in Theorem \ref{theorem: enum algorithm}. Theorem \ref{thm:enumhkcut-algo} follows from Theorem \ref{theorem: enum algorithm} by observing that the source minimal minimum $(S,T)$-terminal cut in a hypergraph can be computed in deterministic polynomial time---e.g., it can be computed in a $n$-vertex hypergraph of size $p$ in $O(np)$ time \cite{ChekuriX18}.  

\begin{figure*}[ht]
\centering\small
\begin{algorithm}
\textul{Algorithm Enum-Cuts$(G=(V,E),k)$}\+
\\{\bf Input:} Hypergraph $G=(V,E)$ and an integer $k\geq 2$
\\{\bf Output:} Family of all minimum $k$-cut-sets in $G$
\\ Initialize $\mathcal{C}\gets\emptyset$, $\mathcal{F}\gets\emptyset$
\\ For each pair $(S,T)$ such that $S,T\subseteq V$ with $S\cap T=\emptyset$ and $|S|,|T|\leq 2k-1$ \+
\\ Compute the source minimal minimum $(S,T)$-terminal cut $(U,\complement{U})$
\\ If $G-\delta(U)$ has at least $k$ connected components \+
\\ $\family\leftarrow \family\cup \set{\delta(U)}$\-
\\ Else \+
\\ $\mathcal{C}\gets\mathcal{C}\cup\{U\}$\-\-
\\ For each $k$-partition $(U_1, \ldots, U_k)$ of $V$ with $U_1, \ldots, U_k\in \mathcal{C}$ \+
\\ $\mathcal{F}\gets\mathcal{F}\cup\{\delta(U_1, \ldots, U_k)\}$\-
\\ Among all $k$-cut-sets in the family $\mathcal{F}$, return the subfamily of cheapest ones
\end{algorithm}
\caption{Algorithm to enumerate hypergraph minimum $k$-cut-sets}
\label{Algo: enum-cuts}
\end{figure*}

\begin{theorem} \label{theorem: enum algorithm}
Let $G=(V,E)$ be a $n$-vertex hypergraph of size $p$ and let $k$ be an integer. Then, Algorithm Enum-Cuts$(G,k)$ in Figure \ref{Algo: enum-cuts} returns the family of all minimum $k$-cut-sets in $G$ and it can be implemented to run in $O(n^{4k-2})T(n,p)+O(n^{4k^2-2k}p)$ time, where $T(n,p)$ denotes the time complexity for computing the source minimal minimum $(s,t)$-terminal cut in a $n$-vertex hypergraph of size $p$. Moreover, the cardinality of the family returned by the algorithm is $O(n^{2k(2k-1)})$.
\end{theorem}

\begin{proof}
We begin by showing correctness. 
The last step of the algorithm considers only $k$-cut-sets in the family $\family$, 
so the algorithm returns a subfamily of $k$-cut-sets. We only have to show that every minimum $k$-cut-set is in the family $\family$; 
this will also guarantee that every $k$-cut-set in the returned subfamily is indeed a minimum $k$-cut-set. 
 
Let $F\subseteq E$ be a minimum $k$-cut-set in $G$ and let $(V_1,\ldots,V_k)$ be a minimum $k$-partition such that $F=\delta(V_1,\ldots,V_k)$. We will show that $F$ is in the family 
$\mathcal{F}$. 
We know that $d(V_i)\le \opt$ for every $i\in [k]$. We distinguish two cases:
\begin{enumerate}
    \item Suppose $d(V_i)<\opt$ for every $i\in [k]$.
    
    Consider an arbitrary part $V_i$ where $i\in [k]$. By Theorem \ref{theorem: structure thm 1}, there exist disjoint subsets $S,T\subseteq V$ with $|S|,|T|\leq 2k-2$ such that $(V_i,\complement{V_i})$ is the unique minimum $(S,T)$-terminal cut. Hence, the set $V_i$ is in the collection $\mathcal{C}$. 
    Consequently, all parts $V_1, \ldots, V_k$ are in the collection $\collection$. Hence, the set $F=\delta(V_1, \ldots, V_k)$ is added to the family $\family$ in the second for-loop. 
   
    \item Suppose there exists $i\in [k]$ such that $d(V_i)=\opt$.
    
    In this case, we have $\delta(V_i)=F=\delta(V_1,\ldots,V_k)$. By Theorem \ref{theorem: structure thm 2}, there exist disjoint subsets $S,T\subseteq V$ with $|S|,|T|\leq 2k-1$ such that the source minimal minimum $(S,T)$-terminal cut $(A,\complement{A})$ satisfies $\delta(A)=\delta(V_i)=F$. 
    Therefore, the set $F$ is added to the family $\family$ in the first for-loop. 
\end{enumerate}
Thus, in both cases, we have shown that the set $F$ is contained in the family $\mathcal{F}$. Since the algorithm returns the subfamily of hyperedge sets in $\mathcal{F}$ that correspond to minimum $k$-cut-sets, the set $F$ is in the family returned by the algorithm.

Next, we bound the run time and the number of minimum $k$-cut-sets returned by the algorithm. The first for-loop can be implemented using $O(n^{4k-2})$ source minimal minimum $(s,t)$-terminal cut computations. Moreover, the size of the collection $\collection$ is $O(n^{4k-2})$. The number of tuples $(U_1, \dots, U_k) \in \mathcal{C}^k$ is $O(n^{4k^2-2k})$. Verifying if a tuple $(U_1, \dots, U_k)$ forms a $k$-partition takes $O(n)$ time. For a tuple which forms a $k$-partition, computing the hyperedges crossing that partition takes $O(p)$ time. Thus, the second for-loop can be implemented to run in time $O(n^{4k^2-2k}p)$.
The size of the family $\family$ is $O(n^{4k^2-2k})$. Each $k$-cut-set in $\family$ has representation size at most $p$. Hence, computing the size of each $k$-cut-set in $\family$ and returning the cheapest ones can be implemented to run in time $O(n^{4k^2-2k}p)$.
Thus, the overall run-time is $O(n^{4k-2})T(n,p)+O(n^{4k^2-2k}p)$.  
\end{proof}

\section{Proof of Theorem \ref{theorem: structure thm 1}} \label{sec:structure-1}

We prove Theorem \ref{theorem: structure thm 1} in this section. 
We will use the following theorem to prove Theorem \ref{theorem: structure thm 1}.

\begin{theorem}\label{thm: thm for structure thm 1}
Let $G=(V,E)$ be a hypergraph and let $\opt$ be the value of a minimum $k$-cut-set in $G$ for some integer $k\ge 2$. Suppose $(U,\complement{U})$ is a $2$-partition of $V$ with $d(U)<\opt$. Then, for every vertex $s\in U$, there exists a subset $S\subseteq U\backslash\{s\}$ with $|S|\leq 2k-3$ such that $(U,\complement{U})$ is the unique minimum $(S\cup\{s\},\complement{U})$-terminal cut.
\end{theorem}

\begin{proof}
Let $s\in U$. Consider the collection  \[
\mathcal{C}:=\{Q\subseteq V\backslash\{s\}:\complement{U}\subsetneq Q,d(Q)\leq d(U)\}. 
\]
Let $S$ be an inclusion-wise minimal subset of $U\setminus \{s\}$ such that $S\cap Q\neq \emptyset$ for all $Q\in \mathcal{C}$,  i.e., the set $S$ is completely contained in $U\backslash\{s\}$ and is a minimal transversal of $\mathcal{C}$.
Proposition \ref{prop:unique-min-cut} and Lemma  \ref{lemma:size-of-witness} complete the proof of Theorem  \ref{thm: thm for structure thm 1} for this choice of $S$.
\end{proof}

\begin{proposition}\label{prop:unique-min-cut}
The $2$-partition $(U,\complement{U})$ is the unique minimum $(S\cup\{s\},\complement{U})$-terminal cut.
\end{proposition}
\begin{proof}
For the sake of contradiction, suppose $(Y,\complement{Y})$ is a minimum $(S\cup\{s\},\complement{U})$-terminal cut with $Y\neq U$. This implies that $S\cup \{s\}\subseteq Y$ and  $\complement{U}\subsetneq\complement{Y}$. Moreover, we have $d(\complement{Y})\leq d(\complement{U})$ because $(U,\complement{U})$ is a $(S\cup\{s\},\complement{U})$-terminal cut. Consequently, the set $\complement{Y}$ is in the collection $\mathcal{C}$. Since $S$ is a transversal of the collection $\mathcal{C}$, we have that $S\cap\complement{Y}\neq\emptyset$. This contradicts  the fact that $S$ is contained in $Y$.
\end{proof}

\begin{lemma}\label{lemma:size-of-witness}
The size of the subset $S$ is at most $2k-3$.
\end{lemma}
\begin{proof}
For the sake of contradiction, suppose $|S|\geq 2k-2$. Our proof strategy is to show the existence of a $k$-partition with cost smaller than $\opt$, thus contradicting the definition of $\opt$. 
Let $S:=\{u_1,u_2,\ldots,u_p\}$ for some $p\geq 2k-2$. For each $i\in[p]$, let $(\complement{A_i},A_i)$ be the source minimal minimum $((S\cup\{s\})\backslash\{u_i\},\complement{U})$-terminal cut. The following claim will allow us to show that the cuts $(\complement{A_i}, A_i)$ satisfy the hypothesis of Theorem \ref{theorem:hypergraph-uncrossing}. 

\begin{claim}\label{claim:min-separating-cuts-miss-u_i}
For every $i\in[p]$, we have $d(A_i)\leq d(U)$ and $u_i\in A_i$.
\end{claim}

\begin{proof}
Let $i\in [p]$. Since $S$ is a minimal transversal of the collection $\mathcal{C}$, there exists a set $B_i\in \mathcal{C}$ such that $B_i\cap S = \{u_i\}$. Hence, $(\complement{B_i},B_i)$ is a $((S\cup\{s\})\backslash\{u_i\},\complement{U})$-terminal cut. Therefore, 
\[
d(A_i)\leq d(B_i)\leq d(U). 
\]

We will show that $A_i$ is in the collection $\mathcal{C}$. By definition, $A_i\subseteq V\setminus \{s\}$ and $\complement{U}\subseteq A_i$. If $A_i=\complement{U}$, then the above inequalities are equations implying that $(B_i, \complement{B_i})$ is a minimum  $((S\cup\{s\})\backslash\{u_i\}, \complement{U})$-terminal cut, and consequently, $(B_i, \complement{B_i})$ contradicts source minimality of the minimum $((S\cup\{s\})\backslash\{u_i\}, \complement{U})$-terminal cut $(A_i, \complement{A_i})$. Therefore, $\complement{U}\subsetneq A_i$. Hence, $A_i$ is in the collection $\mathcal{C}$. 

We recall that the set $S$ is a transversal for the collection $\mathcal{C}$ and moreover, none of the elements of $S\setminus \{u_i\}$ are in $A_i$  by definition of $A_i$. Therefore, the vertex $u_i$ must be in $A_i$.  
\end{proof}

Using Claim \ref{claim:min-separating-cuts-miss-u_i}, we observe that the sets $U$, $R:=\set{s}$, $S$, and the partitions $(\complement{A_i}, A_i)$ for $i\in [p]$ satisfy the conditions of Theorem \ref{theorem:hypergraph-uncrossing}. By the first conclusion of Theorem \ref{theorem:hypergraph-uncrossing} and Claim \ref{claim:min-separating-cuts-miss-u_i}, we obtain a $k$-partition $(P_1, \ldots, P_k)$ of $V$ such that 
\begin{align*}
\cost(P_1, \ldots, P_k)
&\le \frac{1}{2}\min\set{\deltacard(A_i) + \deltacard(A_j):i, j\in [p], i\neq j}
\le \deltacard(U)
< \opt.
\end{align*}
The last inequality above is by the assumption in the theorem statement. Thus, we have obtained a $k$-partition whose cost is smaller than $\opt$, a contradiction. 

\end{proof}

Applying Theorem \ref{thm: thm for structure thm 1} to $(\complement{U}, U)$ yields the following corollary.

\begin{corollary}\label{corollary: coro for structure thm 1}
Let $G=(V,E)$ be a hypergraph and let $\opt$ be the value of a minimum $k$-cut-set in $G$ for some integer $k\ge 2$. Suppose $(U,\complement{U})$ is a $2$-partition of $V$ with $d(U)<\opt$. Then, for every vertex $t\in \complement{U}$, there exists a subset $T\subseteq \complement{U}\backslash\{t\}$ with $|T|\leq 2k-3$ such that $(U,\complement{U})$ is the unique minimum $(U,T\cup\{t\})$-terminal cut.
\end{corollary}

We now restate Theorem \ref{theorem: structure thm 1} and prove it using Theorem \ref{thm: thm for structure thm 1} and Corollary \ref{corollary: coro for structure thm 1}.

\thmStructureOne*
\begin{proof}
Let $s\in U$ and $t\in \complement{U}$.
By Theorem \ref{thm: thm for structure thm 1}, there exists a subset $S\subseteq U\setminus \{s\}$ such that $|S|\leq 2k-3$ and $(U,\complement{U})$ is the unique minimum $(S\cup\{s\},\complement{U})$-terminal cut. By Corollary \ref{corollary: coro for structure thm 1}, there exists a subset $T\subseteq \complement{U}$ such that $|T|\leq 2k-3$ and $(U,\complement{U})$ is the unique minimum $(U,T\cup\{t\})$-terminal cut. 
We now show that 
$(U,\complement{U})$ is the unique minimum $(S\cup\{s\},T\cup\{t\})$-terminal cut.

We now show that $(U, \complement{U})$ is the unique minimum $(S\cup\{s\},T\cup\{t\})$-terminal cut. 
Let $(Y,\complement{Y})$ be a minimum $(S\cup\{s\},T\cup\{t\})$-terminal cut. Suppose $Y\neq U$. 
We have the following observations:
\begin{enumerate}
    \item Since $(U,\complement{U})$ is a $(S\cup\{s\},T\cup\{t\})$-terminal cut, we have that $d(U)\geq d(Y)$.
    \item Since $(U\cap Y,\complement{U\cap Y})$ is a $(S\cup\{s\},\complement{U})$-terminal cut, we have that $d(U\cap Y)\ge d(U)$. 
    \item Since $(U\cup Y,\complement{U\cup Y})$ is a $(U,T\cup\{t\})$-terminal cut, we have that  $d(U\cup Y)\ge d(U)$.
\end{enumerate}
Moreover, since $Y\neq U$, we have that either $U\cap Y\neq U$ or $U\cup Y\neq U$. Since $(U,\complement{U})$ is the unique minimum $(S\cup\{s\},\complement{U})$-terminal cut and also the unique minimum $(U,T\cup\{t\})$-terminal cut, it follows that 
either $d(U\cap Y)> d(U)$ or $d(U\cup Y)> d(U)$. 
These observations in conjunction with the submodularity of the hypergraph cut function imply that 
\begin{align*}
    2d(U)\geq d(U)+d(Y)\geq d(U\cap Y)+d(U\cup Y)> 2d(U), 
\end{align*}
a contradiction. Hence, $Y=U$. 
\end{proof}

\section{Proof of Theorem \ref{theorem: structure thm 2}} \label{sec:structure-2}

We prove Theorem \ref{theorem: structure thm 2} in this section. We begin with the following useful containment lemma. Variants of this containment lemma have appeared in the literature before under slightly different hypothesis (e.g., see \cite{DJPSY94, GH94, OFN12, CC20}). 

\begin{lemma}\label{lem:Hs_subset_V1}
Let $G=(V,E)$ be a hypergraph, $k \geq 2$ be an integer, $\mathcal{P}=(V_1, \ldots, V_k)$ be a minimum $k$-partition such that $\delta(\mathcal{P}) = \delta(V_1)$, and $S \subseteq V_1$, $T \subseteq \overline{V_1}$ such that $T \cap V_j \neq \emptyset$ for all $j \in \{2, 3, \dots, k\}$. Suppose that $(U, \complement{U})$ is the source minimal minimum $(S,T)$-terminal cut. Then, $U \subseteq V_1$ and $(U, \complement{U})$ is a minimum $(S, \overline{V_1})$-terminal cut.
\end{lemma}

\begin{proof}
We note that $S \subseteq U \cap V_1$, so $(U \cap V_1, \overline{U \cap V_1})$ is a $(S,T)$-terminal cut. Thus, we have 
\begin{equation}\label{eqn:Hs_cap_V1_bound}
    d(U \cap V_1) \geq d(U).
\end{equation}

\begin{figure}[htb]
\centering
\includegraphics[width=0.5\textwidth]{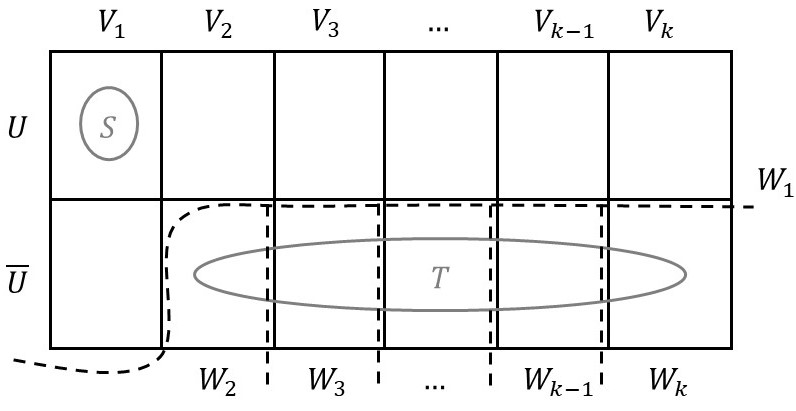}
\caption{Uncrossing in the proof of Lemma \ref{lem:Hs_subset_V1}. }
\label{figure:uncrossing-for-V_S-inside-V_1}
\end{figure}

Now consider $\mathcal{P}' = (W_1:=U \cup V_1, W_2:=V_2 \setminus U, \ldots, W_k:=V_k \setminus U)$ (see Figure \ref{figure:uncrossing-for-V_S-inside-V_1}). For each $i \in \{2, 3, \ldots, k\}$, we have $\emptyset \neq T \cap V_i \subseteq V_i \setminus U$, so $\mathcal{P}'$ is a $k$-partition. Since $\delta(\mathcal{P}) = \delta(V_1)$, every hyperedge which crosses $\mathcal{P}$ must intersect $V_1$. Consequently, every hyperedge which crosses $\mathcal{P}'$ must intersect $U \cup V_1$. Therefore 
\begin{equation}\label{eqn:Hs_cup_V1_bound}
    d(U \cup V_1) =\cost(\mathcal{P}') \geq \cost(\mathcal{P}) = d(V_1).
\end{equation}
By submodularity of the hypergraph cut function and inequalities (\ref{eqn:Hs_cap_V1_bound}) and (\ref{eqn:Hs_cup_V1_bound}), we have that 
\[
d(U) + d(V_1) \geq d(U \cap V_1) + d(U \cup V_1) \geq d(U) + d(V_1).
\]
Therefore, inequality (\ref{eqn:Hs_cap_V1_bound}) is in fact an equation and hence, $(U \cap V_1, \overline{U \cap V_1})$ is a minimum $(S,T)$-terminal cut. 
If $U\setminus V_1\neq \emptyset$, then $(U\cap V_1, \complement{U \cap V_1})$ contradicts source minimality of the minimum $(S,T)$-terminal cut $(U, \complement{U})$. Hence, $U\setminus V_1=\emptyset$ and consequently, $U\subseteq V_1$.

Since $U \subseteq V_1$, we have that $(U, \overline{U})$ is a $(S, \overline{V_1})$-terminal cut. Furthermore, since $T \subseteq \overline{V_1}$, every $(S,\overline{V_1})$-terminal cut is also a $(S,T)$-terminal cut. Therefore, every $(S, \overline{V_1})$-terminal cut must have weight at least $d(U)$, and hence $(U, \overline{U})$ is a minimum $(S,\overline{V_1})$-terminal cut.
\end{proof}

We now restate and prove Theorem \ref{theorem: structure thm 2}. 
\thmStructureTwo*
\begin{proof}
Let us fix an arbitrary $T\subseteq \overline{V_1}$ such that $T \cap V_j \neq \emptyset$ for all $j \in \{2, \dots, k\}$. 
For a subset $X \subseteq V_1$, we denote the source minimal minimum $(X,T)$-terminal cut by $(H_X,\overline{H_X})$. By Lemma \ref{lem:Hs_subset_V1}, for all $X\subseteq V_1$ we have that $H_X\subseteq V_1$ and $d(H_X)\le d(V_1)$. If $|V_1|\le 2k-1$, then choosing $S=V_1$ proves the theorem. So, we will assume henceforth that $|V_1| > 2k-1$.
We will show that there exists a subset $S\subseteq V_1$ with $|S|\le 2k-1$ such that the source minimal minimum $(S,T)$-terminal cut $(H_S, \complement{H_S})$ satisfies $\delta(H_S)=\delta(V_1)$. This suffices since we have that $H_S\subseteq V_1$ for all subsets $S\subseteq V_1$ (by Lemma \ref{lem:Hs_subset_V1}).

For the sake of contradiction, suppose that 
for every $S\subseteq V_1$ with $|S|\le 2k-1$, the source minimal minimum $(S,T)$-terminal cut $(H_S, \complement{H_S})$ does not satisfy $\delta(H_S)=\delta(V_1)$. Our proof strategy is to obtain a cheaper $k$-partition than $(V_1, \ldots, V_k)$, thereby contradicting the optimality of $(V_1, \ldots, V_k)$. 

Let $S \subseteq V_1$ be a set of size $2k-1$ such that $H_S$ is maximal---i.e., there does not exist $S'\subseteq V_1$ of size $2k-1$ such that $H_{S'}\supsetneq H_S$. 
Let $S := \{u_1, u_2, \dots, u_{2k-1}\}$.
By assumption, we have that $\delta(H_S) \neq \delta(V_1)$, but since $(V_1, \complement{V_1})$ is a $(S,T)$-terminal cut, we have that $d(H_S) \leq d(V_1)$. Therefore, $\delta(V_1) \setminus \delta(H_S)$ is non-empty. Let us fix a hyperedge $e \in \delta(V_1) \setminus \delta(H_S)$. Let $u_{2k} \in e \cap V_1$. Let $C := S \cup \{u_{2k}\}=\{u_1, \ldots, u_{2k-1}, u_{2k}\}$. For notational convenience we will use $C-u_i$ to denote $C \setminus \{u_i\}$ and $C-u_i-u_j$ to denote $C \setminus \{u_i,u_j\}$ for all $i, j\in [2k]$. The choice of the hyperedge $e$ is crucial to our proof---its properties will be used much later in our proof. We summarize the properties of the hyperedge $e$ here. 

\begin{observation}\label{obs:properties_of_e}
The hyperedge $e$ has the following properties: 
\begin{enumerate}
    \item $e \cap \overline{V_1} \neq \emptyset$, 
    \item $u_{2k} \in e$, and
    \item $e \subseteq \overline{H_S}$. 
\end{enumerate}
\end{observation}

Our strategy to arrive at a cheaper $k$-partition than $(V_1, \ldots, V_k)$ is to apply the second conclusion of Theorem \ref{theorem:hypergraph-uncrossing}. The next few claims will set us up to obtain sets that satisfy the hypothesis of Theorem \ref{theorem:hypergraph-uncrossing}.
\begin{claim}\label{clm:ui_not_in_Ci}
For every $i \in [2k]$, we have $u_i \not\in H_{C - u_i}$.
\end{claim}
\begin{proof}
If $i = 2k$, then by Observation \ref{obs:properties_of_e} we have $u_{2k} \in e$ and $e \subseteq \overline{H_S}$ so $u_{2k} \not\in H_S = H_{C - u_{2k}}$. Suppose $i \in [2k-1]$. Our proof will rely on the choice of $S$.

Suppose for contradiction that $u_i \in H_{C - u_i}$ for some $i \in [2k-1]$. Then, we have that $S \subseteq H_{C - u_i}$, so $(H_{C - u_i} \cap H_S, \overline{H_{C - u_i} \cap H_S})$ 
is a $(S,T)$-terminal cut. Therefore, 
\begin{equation}\label{eqn:ui_not_in_Ci_eq1}
    d(H_{C - u_i} \cap H_S) \geq d(H_S).
\end{equation}
Also, since $(H_{C - u_i} \cup H_S, \overline{H_{C - u_i} \cup H_S})$ is a $(C - u_i, T)$-terminal cut, we have that
\begin{equation}\label{eqn:ui_not_in_Ci_eq2}
    d(H_{C - u_i} \cup H_S) \geq d(H_{C - u_i}).
\end{equation}
By submodularity of the hypergraph cut function and inequalities (\ref{eqn:ui_not_in_Ci_eq1}) and (\ref{eqn:ui_not_in_Ci_eq2}), we have that 
\[
d(H_S) + d(H_{C - u_i}) \geq d(H_{C - u_i} \cap H_S) + d(H_{C - u_i} \cup H_S) \geq d(H_S) + d(H_{C - u_i}).
\]
Therefore, inequality (\ref{eqn:ui_not_in_Ci_eq1}) is an equation, and consequently,  $(H_{C - u_i} \cap H_S, \overline{H_{C - u_i} \cap H_S})$ is a minimum $(S,T)$-terminal cut. If $H_{C - u_i} \cap H_S \subsetneq H_S$, then $(H_{C - u_i} \cap H_S, \overline{H_{C - u_i} \cap H_S})$ contradicts source minimality of the minimum $(S,T)$-terminal cut $(H_S, \complement{H_S})$. Therefore $H_{C - u_i}\cap H_S=H_S$ and hence, $H_S \subseteq H_{C - u_i}$. Also, the vertex $u_{2k}$ is in $C - u_i$ but not in  $H_S$ and hence, $H_S \subsetneq H_{C - u_i}$. However, $|C - u_i| = 2k-1$. Therefore, the set $C - u_i$ contradicts the choice of $S$.
\end{proof}

The following claim will help in showing that $u_i, u_j\not\in H_{C-u_i-u_j}$, which in turn, will be used to show that the hypothesis of Theorem \ref{theorem:hypergraph-uncrossing} is satisfied by suitably chosen sets.
\begin{claim}\label{clm:Cij_subset_Ci}
For every $i,j \in [2k]$, we have $H_{C - u_i-u_j} \subseteq H_{C - u_i}$.
\end{claim}
\begin{proof}
We may assume that $i\neq j$. 
We note that $(H_{C - u_i-u_j} \cap H_{C - u_i}, \overline{ H_{C - u_i-u_j} \cap H_{C - u_i } })$ is a $(C - u_i-u_j, T)$-terminal cut. Therefore
\begin{equation}\label{eqn:Cij_subset_Ci_eq1}
    d(H_{C - u_i-u_j} \cap H_{C - u_i}) \geq d(H_{C - u_i-u_j}).
\end{equation}
Also, $(H_{C - u_i-u_j} \cup H_{C - u_i}, \overline{ H_{C - u_i-u_j} \cup H_{C - u_i } })$ is a $(C - u_i, T)$-terminal cut. Therefore
\begin{equation}\label{eqn:Cij_subset_Ci_eq2}
    d(H_{C - u_i-u_j} \cup H_{C - u_i}) \geq d(H_{C - u_i}).
\end{equation}
By submodularity of the hypergraph cut function and inequalities (\ref{eqn:Cij_subset_Ci_eq1}) and (\ref{eqn:Cij_subset_Ci_eq2}), we have that 
\begin{align*}
d(H_{C - u_i-u_j}) + d(H_{C - u_i}) &\geq d(H_{C - u_i-u_j} \cap H_{C - u_i}) + d(H_{C - u_i-u_j} \cup H_{C - u_i}) \\ 
&\geq d(H_{C - u_i-u_j}) + d(H_{C - u_i}).
\end{align*}

Therefore, inequality (\ref{eqn:Cij_subset_Ci_eq1}) is an equation, and consequently, $(H_{C - u_i-u_j} \cap H_{C - u_i}, \overline{ H_{C - u_i-u_j} \cap H_{C - u_i } })$ is a minimum $(C - u_i-u_j, T)$-terminal cut. 
If $H_{C - u_i-u_j} \setminus H_{C - u_i} \neq \emptyset$, then \\ $(H_{C - u_i-u_j} \cap H_{C - u_i}, \complement{H_{C - u_i-u_j} \cap H_{C - u_i}})$ contradicts source minimality of the minimum $(C - u_i-u_j,T)$-terminal cut $(H_{C - u_i-u_j}, \complement{H_{C - u_i-u_j}})$. Hence, $H_{C - u_i-u_j} \setminus H_{C - u_i}=\emptyset$ and consequently, $H_{C - u_i-u_j} \subseteq H_{C - u_i}$.
\end{proof}

Claim \ref{clm:Cij_subset_Ci} implies the following Corollary.
\begin{corollary}\label{cor:uij_not_in_Cij}
For every $i \in [2k]$, we have $u_i, u_j \not\in H_{C - u_i- u_j}$.
\end{corollary}
\begin{proof}
By Claim \ref{clm:ui_not_in_Ci}, we have that $u_i \not\in H_{C - u_i}$ and $u_j \not\in H_{C - u_j}$. Therefore, $u_i, u_j \not\in H_{C - u_i} \cap H_{C - u_j}$. By Claim \ref{clm:Cij_subset_Ci}, $H_{C -u_i- u_j} \subseteq H_{C - u_i}$ and $H_{C -u_i- u_j} \subseteq H_{C - u_j}$. Therefore, $H_{C -u_i- u_j} \subseteq H_{C - u_i} \cap H_{C - u_j}$, and thus, $u_i,u_j \not\in H_{C - u_i-u_j}$.
\end{proof}

The next claim will help in controlling the cost of the $k$-partition that we will obtain by applying Theorem \ref{theorem:hypergraph-uncrossing}.
\begin{claim}\label{clm:ci_eq_v1}
For every $i,j \in [2k]$, we have $d(H_{C - u_i}) = d(V_1) = d(H_{C - u_i-u_j})$.
\end{claim}
\begin{proof}
Let $a,b\in [2k]$. 
Since $(V_1, \overline{V_1})$ is a $(C - u_a,T)$-terminal cut, we have that $d(H_{C - u_a}) \leq d(V_1)$.
Since $(H_{C - u_a}, \overline{H_{C - u_a}})$ is a $(C - u_a-u_b,T)$-terminal cut, we have that $d(H_{C - u_a-u_b}) \leq d(H_{C - u_a}) \leq d(V_1)$. Thus, in order to prove the claim, it suffices to show that $d(H_{C - u_a-u_b}) \geq d(V_1)$.

Suppose for contradiction that $d(H_{C - u_a-u_b}) < d(V_1)$. Let $\ell \in [2k] \setminus \{a,b\}$ be an arbitrary element (which exists since $k\ge 2$). Let $R:= \{u_{\ell}\}$, $U:=V_1$, $S' := C - u_a - u_{\ell}$, and $A_i := \overline{H_{C -u_a-u_i}}$ for every $i \in [2k] \setminus \{a,\ell\}$. 
We note that $|S'|= 2k-2$. 
By Lemma \ref{lem:Hs_subset_V1}, we have that $(\complement{A_i}, A_i)$ is a minimum $(C -u_a -u_i, \overline{V_1})$-terminal cut for every $i \in [2k] \setminus \{a,\ell\}$. 
Moreover, 
by Corollary \ref{cor:uij_not_in_Cij}, 
we have that 
$u_i\in A_i\setminus (\cup_{j\in [2k]\setminus \{a, i, \ell\}}A_j)$ 
for every $i \in [2k] \setminus \{a,\ell\}$. 
Hence, the sets $U$, $R$, and $S'$, and the cuts $(\complement{A_i}, A_i)$ for $i\in [2k] \setminus \{a, \ell\}$ satisfy the conditions of Theorem \ref{theorem:hypergraph-uncrossing}. Therefore, by the first conclusion of Theorem \ref{theorem:hypergraph-uncrossing}, there exists a $k$-partition $\mathcal{P}'$ with 
\[\cost(\mathcal{P}') \leq \frac{1}{2}\min\{d(H_{C -u_a- u_i})+d(H_{C -u_a-u_j}) \colon i,j \in [2k] \setminus \{a,\ell\} \}.\]
By assumption, $d(H_{C - u_a-u_b}) < d(V_1)$ and $b\in [2k]\setminus \{a, \ell\}$, so $\min\{d(H_{C -u_a-u_i})\colon i \in [2k] \setminus \{a,\ell\} \} < d(V_1)$. Since $(V_1, \overline{V_1})$ is a $(C - u_a-u_i, T)$-terminal cut, we have that $d(H_{C -u_a-u_i}) \leq d(V_1)$ for every $i \in [2k] \setminus \{a,\ell\}$. Therefore,
\[\frac{1}{2}\min\{d(H_{C -u_a- u_i})+d(H_{C -u_a-u_j}) \colon i,j \in [2k] \setminus \{a,\ell\} \} < d(V_1) = \cost(\mathcal{P}).\]
Thus, we have that $\cost(\mathcal{P'}) < \cost(\mathcal{P})$, which is a contradiction, since $\mathcal{P}$ is a minimum $k$-partition.
\end{proof}

The next two claims will help in arguing properties about the hyperedge $e$ which will allow us to use the second conclusion of Theorem \ref{theorem:hypergraph-uncrossing}. In particular, we will need Claim \ref{clm:Cl_subset_CiCj}. The following claim will help in proving Claim \ref{clm:Cl_subset_CiCj}.
\begin{claim}\label{clm:Ci_cap_cup_eq_V1}
For every $i,j \in [2k]$, we have 
\[d(H_{C - u_i} \cap H_{C - u_j}) = d(V_1) = d(H_{C - u_i} \cup H_{C - u_j}).\]
\end{claim}
\begin{proof}
Since $(H_{C - u_i} \cap H_{C - u_j}, \overline{H_{C - u_i} \cap H_{C - u_j}} )$ is a $(C - u_i-u_j, T)$-terminal cut, we have that $d(H_{C - u_i} \cap H_{C - u_j}) \geq d(H_{C - u_i-u_j})$. By Claim \ref{clm:ci_eq_v1}, we have that $d(H_{C - u_i-u_j}) = d(V_1) = d(H_{C - u_i})$. Therefore,
\begin{equation}\label{eq:Ci_cap_cup_eq_V1_eq1}
    d(H_{C - u_i} \cap H_{C - u_j}) \geq d(H_{C - u_i}).
\end{equation}
Since $(H_{C - u_i} \cup H_{C - u_j}, \overline{H_{C - u_i} \cup H_{C - u_j}})$ is a $(C - u_j, T)$-terminal cut, we have that
\begin{equation}\label{eq:Ci_cap_cup_eq_V1_eq2}
    d(H_{C - u_i} \cup H_{C - u_j}) \geq d(H_{C - u_j}).
\end{equation}
By submodularity of the hypergraph cut function and inequalities (\ref{eq:Ci_cap_cup_eq_V1_eq1}) and (\ref{eq:Ci_cap_cup_eq_V1_eq2}), we have that
\[
d(H_{C - u_i}) + d(H_{C - u_j}) \geq d(H_{C - u_i} \cap H_{C - u_j}) + d(H_{C - u_i} \cup H_{C - u_j}) \geq d(H_{C - u_i}) + d(H_{C - u_j}).
\]
Therefore, inequalities (\ref{eq:Ci_cap_cup_eq_V1_eq1}) and (\ref{eq:Ci_cap_cup_eq_V1_eq2}) are equations. Thus, by Claim \ref{clm:ci_eq_v1}, we have that
\[
d(H_{C - u_i} \cap H_{C - u_j}) = d(H_{C - u_i}) = d(V_1),
\]
and
\[
d(H_{C - u_i} \cup H_{C - u_j}) = d(H_{C - u_j}) = d(V_1).
\]
\end{proof}

\begin{claim}\label{clm:Cl_subset_CiCj}
For every $i,j,\ell \in [2k]$ with $i \neq j$, we have $H_{C - u_{\ell}} \subseteq H_{C - u_i} \cup H_{C - u_j}$.
\end{claim}
\begin{proof}
If $\ell = i$ or $\ell = j$ the claim is immediate. Thus, we assume that $\ell\not\in \{i, j\}$.
Let $Q := H_{C -u_{\ell}} \setminus (H_{C - u_i} \cup H_{C - u_j})$. We need to show that $Q = \emptyset$. We will show that $(H_{C - u_{\ell}} \setminus Q, \overline{H_{C - u_{\ell}}\setminus Q})$ is a minimum $(C - u_{\ell}, T)$-terminal cut. Consequently, $Q$ must be empty 
(otherwise, $H_{C - u_{\ell}} \setminus Q \subsetneq H_{C - u_{\ell}}$ and hence, $(H_{C - u_{\ell}} \setminus Q, \overline{H_{C - u_{\ell}}\setminus Q})$ contradicts source minimality of the minimum $(C - u_{\ell}, T)$-terminal cut $(H_{C-u_{\ell}}, \complement{H_{C-u_{\ell}}})$). 

We now show that $(H_{C - u_{\ell}} \setminus Q, \overline{H_{C - u_{\ell}}\setminus Q})$ is a minimum $(C - u_{\ell}, T)$-terminal cut. 
Since $H_{C - u_{\ell}} \setminus Q = H_{C - u_{\ell}} \cap (H_{C - u_i} \cup H_{C - u_j})$, we have that $C -u_i-u_j-u_{\ell} \subseteq H_{C - u_{\ell}} \setminus Q$. We also know that $u_i$ and $u_j$ are contained in both $H_{C - u_{\ell}}$ and $H_{C - u_i} \cup H_{C - u_j}$. Therefore, $C - u_{\ell} \subseteq H_{C - u_{\ell}} \setminus Q$. Thus, $(H_{C - u_{\ell}} \setminus Q, \overline{H_{C - u_{\ell}} \setminus Q})$ is a $(C - u_{\ell}, T)$-terminal cut. Therefore, 
\begin{equation}\label{eqn:Cl_subset_CiCj_eq1}
  d(H_{C - u_{\ell}} \cap (H_{C - u_i} \cup H_{C - u_j})) = d(H_{C - u_{\ell}} \setminus Q) \geq d(H_{C - u_{\ell}}).  
\end{equation}
We also have that $(H_{C - u_{\ell}} \cup (H_{C - u_i} \cup H_{C - u_j}), \overline{H_{C - u_{\ell}} \cup (H_{C - u_i} \cup H_{C - u_j})})$ is a $(C - u_i, T)$-terminal cut. Therefore, $d(H_{C - u_{\ell}} \cup (H_{C - u_i} \cup H_{C - u_j})) \geq d(H_{C - u_i})$. By Claims \ref{clm:ci_eq_v1} and \ref{clm:Ci_cap_cup_eq_V1}, we have that $d(H_{C - u_i}) = d(V_1)=d(H_{C - u_i} \cup H_{C - u_j})$. Therefore, 
\begin{equation}\label{eqn:Cl_subset_CiCj_eq2}
    d(H_{C - u_{\ell}} \cup (H_{C - u_i} \cup H_{C - u_j})) \geq  d(H_{C - u_i} \cup H_{C - u_j}).
\end{equation}
By submodularity of the hypergraph cut function and inequalities (\ref{eqn:Cl_subset_CiCj_eq1}) and (\ref{eqn:Cl_subset_CiCj_eq2}), we have that
\begin{align*}
d(H_{C - u_{\ell}}) + d(H_{C - u_i} \cup H_{C - u_j}) 
&\geq d(H_{C - u_{\ell}} \cap (H_{C - u_i} \cup H_{C - u_j})) +  d(H_{C - u_{\ell}} \cup (H_{C - u_i} \cup H_{C - u_j})) \\
&\geq d(H_{C - u_{\ell}}) + d(H_{C - u_i} \cup H_{C - u_j}).
\end{align*}
Therefore, inequalities (\ref{eqn:Cl_subset_CiCj_eq1}) and (\ref{eqn:Cl_subset_CiCj_eq2}) are equations, so $(H_{C - u_{\ell}} \setminus Q, \overline{H_{C - u_{\ell}}\setminus Q})$ is a minimum $(C - u_{\ell}, T)$-terminal cut. 
\end{proof}

Let $R:= \{u_{2k}\}$, $U:=V_1$, and let $(\overline{A_i}, A_i) := (H_{C - u_i}, \overline{H_{C - u_i}})$ for every $i \in [2k-1]$. By Lemma \ref{lem:Hs_subset_V1}, we have that $(\complement{A_i}, A_i)$ is a minimum $(C - u_i, \overline{V_1})$-terminal cut for every $i \in [2k-1]$. Moreover, by Claim \ref{clm:ui_not_in_Ci}, we have that $u_i\in A_i\setminus (\cup_{j\in [2k-1]\setminus \{i\}}A_j)$. Hence, the sets $U$, $R$, and $S$, and the cuts $(\complement{A_i}, A_i)$ for $i\in [2k-1]$ satisfy the conditions of Theorem \ref{theorem:hypergraph-uncrossing}. 
We will use the second conclusion of Theorem \ref{theorem:hypergraph-uncrossing}. We now show that the hyperedge $e$ that we fixed at the beginning of the proof satisfies the conditions mentioned in the second conclusion of Theorem \ref{theorem:hypergraph-uncrossing}. We will use Claim \ref{clm:Cl_subset_CiCj} to prove this.
Let $W:=\cup_{1\le i<j\le 2k-1}(A_i\cap A_j)$ and $Z:=\cap_{i\in [2k-1]} \complement{A_i}$ as in the statement of Theorem \ref{theorem:hypergraph-uncrossing}.

\begin{claim}\label{clm:more_properties_of_e}
The hyperedge $e$ satisfies the following conditions:
\begin{enumerate}
    \item $e \cap W \neq \emptyset$,
    \item $e \cap Z \neq \emptyset$, and
    \item $e \subseteq W \cup Z$.
\end{enumerate}
\end{claim}

\begin{proof}
\begin{enumerate}
    \item 
By Lemma \ref{lem:Hs_subset_V1}, for every $i \in [2k-1]$ we have $\overline{V_1} \subseteq A_i$, and therefore  $\overline{V_1} \subseteq W$. Thus, by Observation \ref{obs:properties_of_e}, we have that $\emptyset \neq e \cap \overline{V_1} \subseteq e \cap W$.

\item By definition, for every $i \in [2k-1]$, we have $u_{2k} \in H_{C - u_i} = \overline{A_i}$, and therefore $u_{2k} \in Z$. Thus, by Observation \ref{obs:properties_of_e}, we have that $e \cap Z \neq \emptyset$.

\item For every $i \in [2k-1]$, let $Y_i := A_i \setminus W$. We note that $(Y_1, \dots, Y_{2k-1}, W,Z)$ is a partition of $V$. Therefore, in order to show that $e \subseteq W \cup Z$, it suffices to show that $e \cap Y_i = \emptyset$ for every $i \in [2k-1]$. By Observation \ref{obs:properties_of_e}, we know that $e\subseteq \complement{H_S}$. We will show that $\complement{H_S}\cap Y_i=\emptyset$ for every $i\in [2k-1]$ which implies that $e\cap Y_i=\emptyset$ for every $i\in [2k-1]$. Let us fix an index $i\in [2k-1]$.
We note that
\begin{align*}
    Y_i 
    &= A_i \setminus W 
    = A_i \setminus \left( \bigcup_{1\le a<b\le 2k-1}(A_a\cap A_b) \right) 
    = A_i \setminus \left( \bigcup_{1\le a<b\le 2k-1}((A_a\cap A_b) \cap A_i) \right) \\
    &= A_i \setminus \left( \bigcup_{j \in [2k-1] \setminus \{i\} }(A_j \cap A_i) \right)
    = A_i \setminus \left( \bigcup_{j \in [2k-1] \setminus \{ i\} } A_j \right) 
    = A_i \cap \overline{\left( \bigcup_{j \in [2k-1] \setminus \{ i\}} A_j \right)} \\
    &= A_i \cap \left( \bigcap_{j \in [2k-1] \setminus \{i\}} \overline{A_j} \right) 
    = \left( \bigcap_{j \in [2k-1] \setminus \{i\}} \overline{A_j} \right) \setminus \overline{A_i}.
\end{align*}
Therefore, 
\begin{align*}
    Y_i \cap \overline{H_S} 
    &= \left( \left( \bigcap_{j \in [2k-1] \setminus \{i\}} \overline{A_j} \right) \setminus \overline{A_i} \right) \cap \overline{H_S} 
    = \left( \left( \bigcap_{j \in [2k-1] \setminus \{i\}} \overline{A_j} \right) \setminus \overline{A_i} \right) \setminus H_S \\
    &= \left( \bigcap_{j \in [2k-1] \setminus \{i\}} H_{C - u_j} \right) \setminus \left( H_{C - u_i} \cup H_{C - u_{2k}} \right).
\end{align*}
By Claim \ref{clm:Cl_subset_CiCj}, we have that $H_{C - u_j} \subseteq H_{C - u_i} \cup H_{C - u_{2k}}$ for every $j \in [2k-1] \setminus \{i\}$. Therefore, $H_{C-u_j}\setminus (H_{C - u_i}\cup H_{C - u_{2k}})=\emptyset$ for every $j\in [2k-1]\setminus \{i\}$, and hence, 
\begin{align*}
    \left( \bigcap_{j \in [2k-1] \setminus \{i\}} H_{C - u_j} \right) \setminus \left( H_{C - u_i} \cup H_{C - u_{2k}} \right) = \emptyset.
\end{align*}
Thus, 
we have $Y_i \cap \overline{H_S} = \emptyset$. 
\end{enumerate}
\end{proof}

By Claim \ref{clm:more_properties_of_e}, the hyperedge $e$ satisfies the conditions of the second conclusion of Theorem \ref{theorem:hypergraph-uncrossing}. Therefore, by Theorem \ref{theorem:hypergraph-uncrossing}, there exists a $k$-partition $\mathcal{P'}$ with 
\begin{align*}
\cost(\mathcal{P'}) 
&< \frac{1}{2}\min\{\deltacard(A_i) + \deltacard(A_j): i, j\in [2k-1], i\neq j\}\\
&= d(V_1) \quad \quad \quad \quad \quad \text{(By Claim \ref{clm:ci_eq_v1})}\\
&= \cost(\mathcal{P}). \quad \quad \quad \quad \text{(By assumption of the theorem)}
\end{align*}
Thus, we have obtained a $k$-partition $\mathcal{P}'$ with $\cost(\mathcal{P}') < \cost(\mathcal{P})$, which is a contradiction since $\mathcal{P}$ is a minimum $k$-partition.
\end{proof}

\section{Stronger Structural Theorem for $k=2$}
\label{sec:stronger-structure-2-for-k-equals-2}

We prove the stronger version of Theorem \ref{theorem: structure thm 2} for $k=2$---namely Theorem \ref{theorem: structure thm 2-for-min-cut}---in this section. 
We 
were able to prove Theorem \ref{theorem: structure thm 2-for-min-cut} 
via two more techniques that are different from the one presented in this section---one technique is via a novel three-cut-set-lemma 
while the second technique is via the \emph{canonical decomposition of hypergraphs} \cite{CE80, Fuj83, Cun80, ChekuriX18}. 
It is unclear how to generalize both these techniques to $k\ge 3$. Here, we present a proof of Theorem \ref{theorem: structure thm 2-for-min-cut} that closely resembles the proof of Theorem \ref{theorem: structure thm 2}. 
We will again use the containment lemma (Lemma \ref{lem:Hs_subset_V1}) in our proof. We mention how we obtain the stronger statement relative to Theorem \ref{theorem: structure thm 2} in the proof below. 
We restate and prove Theorem \ref{theorem: structure thm 2-for-min-cut} now.

\thmStructureTwoForMinCut*

\begin{proof}
Let us fix an arbitrary non-empty subset $T\subseteq \complement{V_1}=V_2$. 
For a subset $X \subseteq V_1$, we denote the source minimal minimum $(X,T)$-terminal cut by $(H_X,\overline{H_X})$. By Lemma \ref{lem:Hs_subset_V1}, for all $X \subseteq V_1$ we have that $H_X\subseteq V_1$. If $|V_1|\le 2$, then choosing $S=V_1$ proves the theorem. So, we will assume henceforth that $|V_1|\ge 3$. We will show that there exists a subset $S\subseteq V_1$ with $|S|\le 2$ such that the source minimal minimum $(S,T)$-terminal cut $(H_S, \complement{H_S})$ satisfies $\delta(H_S)=\delta(V_1)$. This suffices since we have that $H_S\subseteq V_1$ for all subsets $S\subseteq V_1$ (by Lemma \ref{lem:Hs_subset_V1}).

We begin with the following useful claim. We note that Claim \ref{clm:HX_eq_V1} crucially relies on the fact that $V_1$ is a part of a minimum cut (i.e, it crucially relies on $k=2$)---it does not hold if $V_1$ is a part of a minimum $k$-partition for $k \geq 3$. 
\begin{claim}\label{clm:HX_eq_V1}
For every $X \subseteq V_1$, we have that $d(H_X) = d(V_1)$.
\end{claim}
\begin{proof}
Since $X \subseteq V_1$ and $T \subseteq V_2$, we have that $(V_1, V_2)$ is a $(X,T)$-terminal cut. Since $(H_X, \overline{H_X})$ is a minimum $(X,T)$-terminal cut, we have that $d(H_X) \leq d(V_1)$. Since $(H_X, \overline{H_X})$ is a cut, and $(V_1, V_2)$ is a minimum cut, we have that $d(H_X) \geq d(V_1)$. Thus, $d(H_X) = d(V_1)$.
\end{proof}

For the sake of contradiction, suppose that 
for every $S \subseteq V_1$ with $|S| \leq 2$, the source minimal minimum $(S,T)$-terminal cut $(H_S, \complement{H_S})$ does not satisfy $\delta(H_S) = \delta(V_1)$. Our proof strategy is to obtain a cheaper cut than $(V_1, V_2)$, thereby contradicting the optimality of $(V_1, V_2)$. 

Let $S \subseteq V_1$ be a set of size $2$ such that $H_S$ is maximal---i.e., there does not exist $S'\subseteq V_1$ of size $2$ such that $H_{S'}\supsetneq H_S$. In contrast to the proof of Theorem \ref{theorem: structure thm 2}, where subsets $S$ of size $3$ had to be considered, here we only consider subsets $S$ of size $2$. We will see that this suffices to arrive at a contradiction.
Let $S:= \{u_1, u_2\}$. By assumption, we have that $\delta(H_S) \neq \delta(V_1)$, but by Claim \ref{clm:HX_eq_V1}, we have that $d(H_S) = d(V_1)$. Therefore, $\delta(V_1) \setminus \delta(H_S)$ is non-empty. Let $e \in \delta(V_1) \setminus \delta(H_S)$. Let $u_{3} \in e \cap V_1$. Let $C := \{u_1, u_2, u_3\}$. For notational convenience we will use $C-u_i$ to denote $C \setminus \{u_i\}$ and $C-u_i-u_j$ to denote $C \setminus \{u_i,u_j\}$ for all $i, j\in [3]$. 
The choice of the hyperedge $e$ is crucial to our proof---its properties will be used much later in our proof. We summarize the properties of the hyperedge $e$ here. 

\begin{observation}\label{obs:properties_of_e_k2}
The hyperedge $e$ has the following properties: 
\begin{enumerate}
    \item $e \cap V_2 \neq \emptyset$
    \item $u_{3} \in e$, and
    \item $e \subseteq \overline{H_S}$,
\end{enumerate}
\end{observation}

Our strategy to arrive at a cheaper cut than $(V_1, V_2)$ is to apply the second conclusion of Theorem \ref{theorem:hypergraph-uncrossing}. The next few claims will set us up to obtain sets that satisfy the hypothesis of Theorem \ref{theorem:hypergraph-uncrossing}.
\begin{claim}\label{clm:ui_not_in_Ci_k2}
For every $i \in [3]$, we have $u_i \not\in H_{C - u_i}$.
\end{claim}
\begin{proof}
By Observation \ref{obs:properties_of_e_k2} we have $u_{3} \in e$ and $e \subseteq \overline{H_S}$, so $u_{3} \not\in H_S = H_{C - u_{3}}$. Suppose $i \in [2]$. Our proof will rely on the choice of $S$.

Suppose for contradiction that $u_i \in H_{C - u_i}$ for some $i \in [2]$. Then we have that $S \subseteq H_{C - u_i}$, so $(H_{C - u_i} \cap H_S, \overline{H_{C - u_i} \cap H_S})$ 
is a $(S,T)$-terminal cut. Therefore, 
\begin{equation}\label{eqn:ui_not_in_Ci_eq1_k2}
    d(H_{C - u_i} \cap H_S) \geq d(H_S).
\end{equation}
Also, since $(H_{C - u_i} \cup H_S, \overline{H_{C - u_i} \cup H_S})$ is a $(C - u_i, T)$-terminal cut, we have that
\begin{equation}\label{eqn:ui_not_in_Ci_eq2_k2}
    d(H_{C - u_i} \cup H_S) \geq d(H_{C - u_i}).
\end{equation}
By submodularity of the hypergraph cut function and inequalities (\ref{eqn:ui_not_in_Ci_eq1_k2}) and (\ref{eqn:ui_not_in_Ci_eq2_k2}), we have that 
\[
d(H_S) + d(H_{C - u_i}) \geq d(H_{C - u_i} \cap H_S) + d(H_{C - u_i} \cup H_S) \geq d(H_S) + d(H_{C - u_i}).
\]
Therefore, inequality (\ref{eqn:ui_not_in_Ci_eq1_k2}) is an equation, and consequently,  $(H_{C - u_i} \cap H_S, \overline{H_{C - u_i} \cap H_S})$ is a minimum $(S,T)$-terminal cut. If $H_{C - u_i} \cap H_S \subsetneq H_S$, then this contradicts the source minimality of the minimum $(S,T)$-terminal cut $(H_S, \complement{H_S})$. Therefore, $H_{C - u_i}\cap H_S=H_S$ and hence, $H_S \subseteq H_{C - u_i}$. Also, the vertex $u_{3}$ is in $C - u_i$ but not in  $H_S$ and hence, $H_S \subsetneq H_{C - u_i}$. However, $|C - u_i| = 2$. Therefore, the set $C - u_i$ contradicts the choice of $S$.
\end{proof}

The next two claims will help in arguing properties about the hyperedge $e$ which will allow us to use the second conclusion of Theorem \ref{theorem:hypergraph-uncrossing}. In particular, we will need Claim \ref{clm:Cl_subset_CiCj_k2}. The following claim will help in proving Claim \ref{clm:Cl_subset_CiCj_k2}. We note the similarity of Claims \ref{clm:Ci_cap_cup_eq_V1_k2} and \ref{clm:Cl_subset_CiCj_k2} in this proof to Claims \ref{clm:Ci_cap_cup_eq_V1} and \ref{clm:Cl_subset_CiCj} in the proof of Theorem \ref{theorem: structure thm 2}. In order to prove Claims \ref{clm:Ci_cap_cup_eq_V1} and \ref{clm:Cl_subset_CiCj}, we needed the size of $C$ to be $2k$. Here, we are able to prove Claims \ref{clm:Ci_cap_cup_eq_V1_k2} and \ref{clm:Cl_subset_CiCj_k2} with the size of $C$ being $2k-1$ (for $k=2$). We do this by exploiting Claim \ref{clm:HX_eq_V1} shown earlier (which holds only for $k=2$).
\begin{claim}\label{clm:Ci_cap_cup_eq_V1_k2}
For every $i,j \in [3]$, we have 
\[d(H_{C - u_i} \cap H_{C - u_j}) = d(V_1) = d(H_{C - u_i} \cup H_{C - u_j}).\]
\end{claim}
\begin{proof}
Since $(H_{C - u_i} \cap H_{C - u_j}, \overline{H_{C - u_i} \cap H_{C - u_j}} )$ is a $(C - u_i-u_j, T)$-terminal cut, we have that $d(H_{C - u_i} \cap H_{C - u_j}) \geq d(H_{C - u_i-u_j})$. By Claim \ref{clm:HX_eq_V1}, we have that $d(H_{C - u_i-u_j}) = d(V_1) = d(H_{C - u_i})$. Therefore,
\begin{equation}\label{eq:Ci_cap_cup_eq_V1_eq1_k2}
    d(H_{C - u_i} \cap H_{C - u_j}) \geq d(H_{C - u_i}).
\end{equation}
Since $(H_{C - u_i} \cup H_{C - u_j}, \overline{H_{C - u_i} \cup H_{C - u_j}})$ is a $(C - u_j, T)$-terminal cut, we have that
\begin{equation}\label{eq:Ci_cap_cup_eq_V1_eq2_k2}
    d(H_{C - u_i} \cup H_{C - u_j}) \geq d(H_{C - u_j}).
\end{equation}
By submodularity of the hypergraph cut function and inequalities (\ref{eq:Ci_cap_cup_eq_V1_eq1_k2}) and (\ref{eq:Ci_cap_cup_eq_V1_eq2_k2}), we have that
\[
d(H_{C - u_i}) + d(H_{C - u_j}) \geq d(H_{C - u_i} \cap H_{C - u_j}) + d(H_{C - u_i} \cup H_{C - u_j}) \geq d(H_{C - u_i}) + d(H_{C - u_j}).
\]
Therefore, inequalities (\ref{eq:Ci_cap_cup_eq_V1_eq1}) and (\ref{eq:Ci_cap_cup_eq_V1_eq2}) are equations. Thus, by Claim \ref{clm:HX_eq_V1} we have that
\[
d(H_{C - u_i} \cap H_{C - u_j}) = d(H_{C - u_i}) = d(V_1),
\]
and
\[
d(H_{C - u_i} \cup H_{C - u_j}) = d(H_{C - u_j}) = d(V_1).
\]
\end{proof}

The next claim follows from Claim \ref{clm:Ci_cap_cup_eq_V1_k2} similar to the proof of Claim \ref{clm:Cl_subset_CiCj} from Claim \ref{clm:Ci_cap_cup_eq_V1} earlier. We include the proof for the sake of completeness. 

\begin{claim}\label{clm:Cl_subset_CiCj_k2}
For every $i,j,\ell \in [3]$ with $i \neq j$, we have $H_{C - u_{\ell}} \subseteq H_{C - u_i} \cup H_{C - u_j}$.
\end{claim}
\begin{proof}
If $\ell = i$ or $\ell = j$ the claim is immediate. Thus, we assume that $\ell\neq i, j$.
Let $Q := H_{C -u_{\ell}} \setminus (H_{C - u_i} \cup H_{C - u_j})$. We need to show that $Q = \emptyset$. We will show that $(H_{C - u_{\ell}} \setminus Q, \overline{H_{C - u_{\ell}}\setminus Q})$ is a minimum $(C - u_{\ell}, T)$-terminal cut. 
Consequently, $Q$ must be empty 
(otherwise, $H_{C - u_{\ell}} \setminus Q \subsetneq H_{C - u_{\ell}}$ and hence, $(H_{C - u_{\ell}} \setminus Q, \overline{H_{C - u_{\ell}}\setminus Q})$ contradicts source minimality of the minimum $(C - u_{\ell}, T)$-terminal cut $(H_{C-u_{\ell}}, \complement{H_{C-u_{\ell}}})$). 

We now show that $(H_{C - u_{\ell}} \setminus Q, \overline{H_{C - u_{\ell}}\setminus Q})$ is a minimum $(C - u_{\ell}, T)$-terminal cut. 
Since $H_{C - u_{\ell}} \setminus Q = H_{C - u_{\ell}} \cap (H_{C - u_i} \cup H_{C - u_j})$, we have that $C -u_i-u_j-u_{\ell} \subseteq H_{C - u_{\ell}} \setminus Q$. We also know that $u_i$ and $u_j$ are contained in both $H_{C - u_{\ell}}$ and $H_{C - u_i} \cup H_{C - u_j}$. Therefore, $C - u_{\ell} \subseteq H_{C - u_{\ell}} \setminus Q$. Thus, $(H_{C - u_{\ell}} \setminus Q, \overline{H_{C - u_{\ell}} \setminus Q})$ is a $(C - u_{\ell}, T)$-terminal cut. Therefore, 
\begin{equation}\label{eqn:Cl_subset_CiCj_eq1_k2}
  d(H_{C - u_{\ell}} \cap (H_{C - u_i} \cup H_{C - u_j})) = d(H_{C - u_{\ell}} \setminus Q) \geq d(H_{C - u_{\ell}}).  
\end{equation}
We also have that $(H_{C - u_{\ell}} \cup (H_{C - u_i} \cup H_{C - u_j}), \overline{H_{C - u_{\ell}} \cup (H_{C - u_i} \cup H_{C - u_j})})$ is a $(C - u_i, T)$-terminal cut. Therefore, $d(H_{C - u_{\ell}} \cup (H_{C - u_i} \cup H_{C - u_j})) \geq d(H_{C - u_i})$. By Claims \ref{clm:HX_eq_V1} and \ref{clm:Ci_cap_cup_eq_V1_k2}, we have that $d(H_{C - u_i}) = d(V_1)=d(H_{C - u_i} \cup H_{C - u_j})$. Therefore, 
\begin{equation}\label{eqn:Cl_subset_CiCj_eq2_k2}
    d(H_{C - u_{\ell}} \cup (H_{C - u_i} \cup H_{C - u_j})) \geq  d(H_{C - u_i} \cup H_{C - u_j}).
\end{equation}
By submodularity of the hypergraph cut function and inequalities (\ref{eqn:Cl_subset_CiCj_eq1_k2}) and (\ref{eqn:Cl_subset_CiCj_eq2_k2}), we have that
\begin{align*}
d(H_{C - u_{\ell}}) + d(H_{C - u_i} \cup H_{C - u_j}) 
&\geq d(H_{C - u_{\ell}} \cap (H_{C - u_i} \cup H_{C - u_j})) +  d(H_{C - u_{\ell}} \cup (H_{C - u_i} \cup H_{C - u_j})) \\
&\geq d(H_{C - u_{\ell}}) + d(H_{C - u_i} \cup H_{C - u_j}).
\end{align*}
Therefore, inequalities (\ref{eqn:Cl_subset_CiCj_eq1_k2}) and (\ref{eqn:Cl_subset_CiCj_eq2_k2}) are equations, so $(H_{C - u_{\ell}} \setminus Q, \overline{H_{C - u_{\ell}}\setminus Q})$ is a minimum $(C - u_{\ell}, T)$-terminal cut. 
\end{proof}

Let $R:= \{u_{3}\}$, $U:=V_1$, and for every $i \in [2]$, let $(\overline{A_i}, A_i) := (H_{C - u_i}, \overline{H_{C - u_i}})$. By Lemma \ref{lem:Hs_subset_V1}, we have that $(\complement{A_i}, A_i)$ is a minimum $(C - u_i, V_2)$-terminal cut for every $i \in [2]$. Moreover, by Claim \ref{clm:ui_not_in_Ci_k2}, we have that  
$u_i \in A_i \setminus A_{3 - i}$.
Hence, the sets $U$, $R$, and $S$, and the cuts $(\complement{A_i}, A_i)$ for $i\in [2]$ satisfy the conditions of Theorem \ref{theorem:hypergraph-uncrossing}. 
We will use the second conclusion of Theorem \ref{theorem:hypergraph-uncrossing}. We now show that the hyperedge $e$ that we fixed at the beginning of the proof satisfies the conditions mentioned in the second conclusion of Theorem \ref{theorem:hypergraph-uncrossing}. We will use Claim \ref{clm:Cl_subset_CiCj_k2} to prove this.
Let $W:=A_1\cap A_2$ and $Z:= \complement{A_1} \cap \complement{A_2}$ as in the statement of Theorem \ref{theorem:hypergraph-uncrossing}.

\begin{claim}\label{clm:more_properties_of_e_k2}
The hyperedge $e$ satisfies the following conditions:
\begin{enumerate}
    \item $e \cap W \neq \emptyset$,
    \item $e \cap Z \neq \emptyset$, and
    \item $e \subseteq W \cup Z$.
\end{enumerate}
\end{claim}

\begin{proof}
\begin{enumerate}
    \item 
By Lemma \ref{lem:Hs_subset_V1}, for every $i \in [2]$ we have $V_2 \subseteq A_i$, and therefore  $V_2 \subseteq W$. Thus, by Observation \ref{obs:properties_of_e_k2}, we have that $\emptyset \neq e \cap V_2 \subseteq e \cap W$.

\item By definition, for every $i \in [2]$, we have $u_{3} \in H_{C - u_i} = \overline{A_i}$, and therefore $u_{3} \in Z$. Thus, by Observation \ref{obs:properties_of_e_k2}, we have that $e \cap Z \neq \emptyset$.

\item For every $i \in [2]$, let $Y_i := A_i \setminus W$. We note that $(Y_1, Y_2, W,Z)$ is a partition of $V$. Therefore, in order to show that $e \subseteq W \cup Z$, it suffices to show that $e \cap Y_1 = e \cap Y_2 = \emptyset$. By Observation \ref{obs:properties_of_e_k2}, we know that $e\subseteq \complement{H_S}$. We will show that $\complement{H_S}\cap Y_i=\emptyset$ for every $i\in [2]$ which implies that $e\cap Y_i=\emptyset$ for every $i\in [2]$. Let us fix a $i\in [2]$.
We note that
\[
Y_i = A_i \setminus W 
    = A_i \setminus (A_1 \cap A_2) 
    = A_i \setminus A_{3 - i} 
    = A_i \cap \overline{A_{3 -i}} 
    = \overline{A_{3-i}} \setminus \overline{A_i}.
\]
Therefore, 
\[
Y_i \cap \overline{H_S} = \left(\overline{A_{3-i}} \setminus \overline{A_i}\right) \cap \overline{H_S} = \left(\overline{A_{3-i}} \setminus \overline{A_i}\right) \setminus H_S = H_{C - u_{3-i}} \setminus \left( H_{C - u_i} \cup H_{C - u_3} \right).
\]
By Claim \ref{clm:Cl_subset_CiCj_k2}, we have that $H_{C - u_{3-i}} \subseteq H_{C - u_i} \cup H_{C - u_{3}}$. Therefore, 
\[
H_{C - u_{3-i}} \setminus \left( H_{C - u_i} \cup H_{C - u_3} \right) = \emptyset.
\]
Thus, for every $i \in [2]$, we have $Y_i \cap \overline{H_S} = \emptyset$. 
\end{enumerate}
\end{proof}

By Claim \ref{clm:more_properties_of_e_k2}, the hyperedge $e$ satisfies the conditions of the second conclusion of Theorem \ref{theorem:hypergraph-uncrossing}. Therefore, by Theorem \ref{theorem:hypergraph-uncrossing}, there exists a cut $(V'_1, V'_2)$ with 
\begin{align*}
\cost(V_1', V_2') = d(V'_1)
&< \frac{1}{2}(d(A_1) + d(A_2))
= d(V_1). 
\end{align*}
The last equality above is by Claim \ref{clm:HX_eq_V1}. 
Thus, we have obtained a cut $(V_1', V_2')$ with $d(V'_1) < d(V_1)$, which is a contradiction since $(V_1, V_2)$ is a minimum cut.
\end{proof}

\section{Conclusion and Open Problems} 
\label{sec:conclusion}
Several works in the literature have approached global cut and partitioning problems via minimum $(S,T)$-terminal cuts (e.g., see \cite{GH94, GR95, NSZ19, CC20, CC21-minmax-symsm-conf}). 
Our work adds to this rich literature by showing that \enumhkcut can be solved via minimum $(S, T)$-terminal cuts. As a special case, our approach leads to a more straightforward approach to enumerate all minimum cut-sets in a given hypergraph. 
We mention a couple of open questions raised by our work. 
\begin{enumerate}
\item For a long time, the known upper bound on the number of minimum $k$-partitions in connected graphs was $O(n^{2k-2})$ \cite{KS96, Th08, CQX20} while the known lower bound was $\Omega(n^k)$ (cycle), where $n$ is the number of vertices in the input graph. A recent result improved the upper bound to $O(n^k)$ which also resulted in a faster randomized algorithm to solve \gkcut \cite{GLL20-STOC, GHLL20}. We currently know that the number of minimum $k$-cut-sets in hypergraphs is $O(n^{2k-2})$ and is $\Omega(n^k)$ (the lower bound comes from graphs). Can we improve the upper bound on the number of minimum $k$-cut-sets in hypergraphs to $O(n^k)$?  

\item Can we improve the deterministic run-time to solve \enumhkcut? 
Our deterministic algorithm for \enumhkcut runs in time $n^{O(k^2)}p$, where $p$ is the size of the input hypergraph. In particular, the run-time of our algorithm has a quadratic dependence on $k$ in the exponent of $n$. 
In contrast, the number of optimum solutions is only $O(n^{2k-2})$---i.e., linear dependence on $k$ in the exponent of $n$. 
We note that \enumgkcut as well as \hkcut can be solved deterministically in $n^{O(k)}p$ time \cite{Th08, CQX20, CC20}. 

\end{enumerate}

\mypara{Acknowledgements.} We would like to thank Chandra Chekuri for feedback that helped improve the introductory section of this work. Karthik would like to thank Michel Goemans for clarifying the proof of Theorem 15 in \cite{GR95} via email. 

\bibliographystyle{amsplain}
\bibliography{references}

\appendix

\section{Proof of Theorem \ref{theorem:hypergraph-uncrossing}}\label{section:uncrossing-for-hypergraph-cut-function}

We prove Theorem \ref{theorem:hypergraph-uncrossing} in this section. We will need certain partition uncrossing and partition aggregation results from \cite{CC20} that rely on more careful counting of hyperedges than simply employing the submodularity inequality. 
We begin with some notation that will help in such careful counting---our notation will be identical to the notation in \cite{CC20}. 
Let $(Y_1, \ldots, Y_p, W, Z)$
be a partition of $V$. 
We recall that $\cost(Y_1, \ldots, Y_p, W, Z)$
denotes the number of hyperedges that cross the partition.
We define
the following quantities:
\begin{enumerate}
\item Let
  $\cost(W, Z) := |\{ e \mid e \subseteq W \cup Z, e \cap W \neq \emptyset, e
  \cap Z \neq \emptyset\}|$ be the number of hyperedges contained in
  $W\cup Z$ that intersect both $W$ and $Z$.
    \item Let $\alpha(Y_1, \ldots, Y_p, W, Z)$ be the number of hyperedges that intersect $Z$ and at least two of the sets in $\set{Y_1, \ldots, Y_p, W}$.
    \item Let $\beta(Y_1, \ldots, Y_p, Z)$ be the number of hyperedges that are disjoint from $Z$ but intersect at least two of the sets in $\set{Y_1, \ldots, Y_p}$.
\end{enumerate}
For a partition $(Y_1, \ldots, Y_p, W, Z)$, we will be interested in
the sum of $\cost(Y_1, \ldots, Y_p, W, Z)$ with the 
three 
quantities
defined above which we denote as $\sigma(Y_1, \ldots, Y_p, W, Z)$,
i.e.,
\[
\sigma(Y_1, \ldots, Y_p, W, Z) := \cost(Y_1, \ldots, Y_p, W, Z) + \cost(W, Z) + \alpha(Y_1, \ldots, Y_p, W, Z) + \beta(Y_1, \ldots, Y_p, Z).
\]
The precise interpretation of the quantity $\sigma(Y_1, \ldots, Y_p, W, Z)$ will not be important for our purposes---see \cite{CC20} for the interpretation. 

The following result from \cite{CC20} shows that a collection of sets can be uncrossed to obtain a partition with small $\sigma$-value. 
We note the similarity of the hypothesis of Lemma \ref{lemma:uncrossing} with the hypothesis of Theorem \ref{theorem:hypergraph-uncrossing} and once again, refer to Figure
\ref{figure:uncrossing} for an illustration of the sets that appear in the statement of Lemma \ref{lemma:uncrossing}. 

\begin{lemma}\label{lemma:uncrossing} \cite{CC20}
  Let $G=(V,E)$ be a hypergraph and
  $\emptyset\neq R\subsetneq U\subsetneq V$. Let
  $S=\{u_1,\ldots, u_p\}\subseteq U\setminus R$ for $p\ge 2$. Let
  $(\complement{A_i}, A_i)$ be a minimum
  $((S\cup R)\setminus \set{u_i}, \complement{U})$-terminal
  cut. Suppose that
  $u_i\in A_i\setminus (\cup_{j\in [p]\setminus \set{i}}A_j)$ for
  every $i\in [p]$.  Let
\[
Z:= \cap_{i=1}^p \complement{A_i},\ W:= \cup_{1\le i<j\le p}(A_i \cap A_j),\ \text{and}\ Y_i:=A_i-W\ \forall i\in [p].
\]
Then, $(Y_1, \ldots, Y_p, W, Z)$ is a $(p+2)$-partition of $V$ with
\[
\sigma(Y_1, \ldots, Y_p, W, Z) \le \min\{\deltacard(A_i) + \deltacard(A_j): i, j\in [p], i\neq j\}.
\]
Moreover, if $p=2$, then the above inequality is an equation. 
\end{lemma}

The next lemma from \cite{CC20} will help in aggregating the parts of a $\ell$-partition $\mathcal{P}$ where $\ell\ge 2k$ to a $k$-partition $\mathcal{K}$ while controlling the cost of $\mathcal{K}$. 
\begin{lemma}
\label{lemma:aggregating} \cite{CC20}
Let $G=(V,E)$ be a hypergraph, $k\ge 2$ be an integer, and $(Y_1, \ldots, Y_p, W, Z)$ be a partition of $V$ for some integer $p\ge 2k-2$. Then, there exist distinct $i_1, \ldots, i_{k-1}\in [p]$ such that
\[
2\cost\left(Y_{i_1}, \ldots, Y_{i_{k-1}}, V\setminus (\cup_{j=1}^{k-1}Y_{i_j})\right)
\le
\cost(Y_1, \ldots, Y_p, W, Z) + \alpha(Y_1, \ldots, Y_p, W, Z) + \beta(Y_1, \ldots, Y_p, Z).
\]
\end{lemma}

We now restate and prove Theorem \ref{theorem:hypergraph-uncrossing}. \thmHypergraphUncrossing*
\begin{proof}
For the first conclusion of the theorem, we will use the same proof that appeared in \cite{CC20}. We need the details of this proof to prove the second conclusion of the theorem. 

We begin by proving the first conclusion. By applying Lemma \ref{lemma:uncrossing}, we obtain a $(p+2)$-partition $(Y_1, \ldots, Y_p, W, Z)$ such that
\[
\sigma(Y_1, \ldots, Y_p, W, Z) \le \min\{\deltacard(A_i) + \deltacard(A_j): i, j\in [p], i\neq j\}
\]
and moreover, $\complement{U}\subseteq W$, where $Y_i=A_i-W$ for all $i\in [p]$,  $W=\cup_{1\le i<j\le p}(A_i\cap A_j)$, and $Z=\cap_{i\in [p]} \complement{A_i}$. We recall that $p\ge 2k-2$. Hence, by applying Lemma \ref{lemma:aggregating} to the $(p+2)$-partition $(Y_1, \ldots, Y_p, W, Z)$, we obtain a $k$-partition $(P_1, \ldots, P_k)$ of $V$ such that $W\cup Z\subseteq P_k$ and
\begin{align}
\cost(P_1, \ldots, P_k) 
&\le \frac{1}{2}\left(\cost(Y_1, \ldots, Y_p, W, Z) + \alpha(Y_1, \ldots, Y_p, W, Z)+ \beta(Y_1, \ldots, Y_p, Z)\right) \\
&\le \frac{1}{2}\left(\cost(Y_1, \ldots, Y_p, W, Z) + \cost(W,Z) +\alpha(Y_1, \ldots, Y_p, W, Z) + \beta(Y_1, \ldots, Y_p, Z)\right) \label{eq:extra-cost}\\
&= \frac{1}{2}\sigma(Y_1, \ldots, Y_p, W, Z) \\
&\le \frac{1}{2}\min\{\deltacard(A_i) + \deltacard(A_j): i, j\in [p], i\neq j\}.
\end{align}
We note that $\complement{U}$ is strictly contained in $P_k$ since $\complement{U}\cup Z\subseteq W\cup Z\subseteq P_k$ and $Z$ is non-empty.

We now prove the second conclusion of the theorem. If there exists a hyperedge $e\in E$ such that $e$ intersects $W$, $e$ intersects $Z$,  and $e$ is contained in $W\cup Z$, then $\cost(W,Z)>0$. Consequently, inequality (\ref{eq:extra-cost}) in the above sequence of inequalities should be strict. 
\end{proof}

\appendix

\end{document}